\documentclass[12pt]{article}

\usepackage[utf8]{inputenc}
\usepackage[T1]{fontenc}
\usepackage{eucal}
\usepackage{microtype}
\usepackage{csquotes}
\usepackage{amsthm}
\usepackage{mathtools}
\usepackage{amssymb,amsfonts,stmaryrd,mathtools}
\usepackage{tikz}
\usetikzlibrary{cd}
\usepackage{enumitem}
\setlist[enumerate,1]{label={\roman*)}}

\usepackage[hyphens]{url}
\usepackage[hyphenbreaks]{breakurl}
\usepackage[unicode=true, pdfusetitle, colorlinks=true, allcolors=blue, bookmarks=false]{hyperref}
\usepackage{authblk} 
\newcommand{\orcid}[1]{}
\usepackage[left=2.5cm,right=2.5cm,top=3cm,bottom=3.5cm]{geometry}

\mathtoolsset{showonlyrefs}
\theoremstyle{plain}
\newtheorem{theorem}{Theorem}
\newtheorem*{theorem*}{Theorem}
\newtheorem{lemma}[theorem]{Lemma}
\newtheorem*{lemma*}{Lemma}
\newtheorem{proposition}[theorem]{Proposition}
\newtheorem*{proposition*}{Proposition}
\newtheorem{corollary}[theorem]{Corollary}  
\newtheorem*{corollary*}{Corollary}
\theoremstyle{definition}

\newtheorem*{definition*}{Definition}

\newtheorem*{example*}{Example}
\theoremstyle{remark}
\newtheorem{remarkx}{Remark}
\newtheorem*{remarkx*}{Remark}
\newenvironment{remark}
  {\pushQED{\qed}\remarkx}
  {\popQED\endremarkx}

\newtheorem*{conjecture*}{Conjecture}

\newtheorem*{problem*}{Problem}

\newcommand*{\RR}{\mathbb{R}}

\newcommand*{\dd}{d}

\newcommand\restr[2]{{
  \left.\kern-\nulldelimiterspace 
  #1 
  \right|_{#2} 
}}

\usepackage[style=ext-numeric,citestyle=numeric-comp,sorting=nyt,sortcites, backend=biber, giveninits=true, maxnames=99, citetracker=true]{biblatex}

\defbibcheck{uncited}{
  \ifciteseen
    {\skipentry}
    {}
}

\bibliography{biblio}
\DeclareFieldFormat[article,periodical]{volume}{\mkbibbold{#1}}
\DeclareFieldFormat[article,periodical]{number}{\mkbibparens{#1}}
\renewbibmacro{in:}{%
  \ifentrytype{article}
    {}
    {\printtext{\bibstring{in}\intitlepunct}}}

\DeclareFieldFormat{issuedate}{#1}
\renewcommand{\jourvoldelim}{\addcomma\space}

\renewbibmacro*{journal+issuetitle}{%
  \usebibmacro{journal}%
  \setunit*{\jourvoldelim}%
  \iffieldundef{series}
    {}
    {\setunit*{\jourserdelim}%
     \printfield{series}%
     \setunit{\servoldelim}}%
  \usebibmacro{volume+number+eid}%
  \setunit{\addcolon\space}%
  \usebibmacro{issue}%
  \setunit{\bibpagespunct}%
  \printfield{pages}%
  \setunit{\space}%
  \usebibmacro{issue+date}%
  \newunit}

\renewbibmacro*{note+pages}{%
  \printfield{note}%
  \newunit}

\DeclareFieldFormat[article,periodical]{date}{\mkbibparens{#1}}

\AtEveryBibitem{
  \ifentrytype{online}{}{
    \ifentrytype{thesis}{}{
        \clearfield{url}
        \clearfield{urldate}
      }
  }
  \ifentrytype{unpublished}{}{ 
  \clearfield{note}
  }
  \clearfield{language}
}

\DeclareSourcemap{
  \maps[datatype=bibtex]{
    \map[overwrite=true]{
      \step[fieldset=urldate, null]
      \step[fieldset=language, null]
      \step[fieldset=address, null]
     \step[fieldset=pagetotal, null]
    }
  }
}

\title{A new perspective on nonholonomic brackets and Hamilton--Jacobi theory}
\author[1,2]{Manuel de León\orcid{0000-0002-8028-2348}}
\author[3]{Manuel Lainz\orcid{0000-0002-2368-5853}}
\author[1]{Asier López-Gordón\orcid{0000-0002-9620-9647}\footnote{Author to whom correspondence should be addressed: \href{mailto:asier.lopez@icmat.es}{asier.lopez@icmat.es}}}
\author[4]{Juan Carlos Marrero\orcid{0000-0002-9620-9647}}
\affil[1]{Instituto de Ciencias Matemáticas (CSIC-UAM-UC3M-UCM) \protect\\
Calle Nicolás Cabrera, 13-15, Campus Cantoblanco, UAM, 28049 Madrid, Spain}
\affil[2]{Real Academia de Ciencias Exactas, Físicas y Naturales\protect\\
Calle Valverde, 22, 28004, Madrid, Spain }
\affil[3]{Departamento de Métodos Cuantitativos, CUNEF,\protect\\
Calle Leonardo Prieto Castro, 2, Ciudad Universitaria, 28040, Madrid, Spain}
\affil[4]{ULL-CSIC Geometría Diferencial y Mecánica Geométrica, Departamento de Matemáticas Estadística e Investigación Operativa and Instituto de Matemáticas y Aplicaciones (IMAULL), University of La Laguna, 
San Cristóbal de La Laguna, Spain}

\date{}

\let\oldemph\emph
\let\emph\textbf

\begin{document}

\maketitle

\begin{abstract}
The nonholonomic dynamics can be described by the so-called nonholonomic bracket on the constrained submanifold, which is a non-integrable modification of the Poisson bracket of the ambient space, in this case, of the canonical bracket on the cotangent bundle of the configuration manifold. This bracket was defined in \cite{Cantrijn1999a,Ibort1999} although there was already some particular and less direct definition. On the other hand, another bracket, also called nonholonomic bracket, was defined using the description of the problem in terms of skew-symmetric algebroids \cite{deLeon2010a,Grabowski2009}. Recently, reviewing two older papers by R.~J.~Eden \cite{Eden1951,Eden1951b}, we have defined a new bracket which we call Eden bracket. In the present paper, we prove that these three brackets coincide. Moreover, the description of the nonholonomic bracket {\sl \`a la} Eden has allowed us to make important advances in the study of Hamilton--Jacobi theory and the quantization of nonholonomic systems. 

\medskip


\medskip


\end{abstract}

\vfill

\tableofcontents

	\maketitle
	
	\section{Introduction}
	
	One of the most important objects in mechanics is the Poisson bracket, which allows us to obtain the evolution of an observable by bracketing it with the Hamiltonian function, or to obtain new conserved quantities of two given ones, using the Jacobi identity satisfied by the bracket. Moreover, the Poisson bracket is fundamental to proceed with the quantization of the system using what Dirac called the analogy principle, also known as the correspondence principle, according to which the Poisson bracket becomes the commutator of the operators associated to the quantized observables.
	
	For a long time, no similar concept existed in the case of nonholonomic mechanical systems, until van der Schaft and Maschke \cite{vanderSchaft1994} introduced a bracket similar to the canonical Poisson bracket, but without the benefit of integrability (see also \cite{Koon1997}). 
	Later, in \cite{Cantrijn1999a,Cantrijn2000} (see also \cite{Ibort1999}), we have developed a geometric and very simple way to define nonholonomic brackets, in the time-dependent as well time-independent cases. Indeed, it is possible to decompose the tangent bundle and the cotangent bundle along the constraint submanifold in two different ways. Both result in that the nonholonomic dynamics can be obtained by projecting the free dynamics. Furthermore, if we evaluate the projections of the Hamiltonian vector fields of two functions on the configuration manifold (after arbitrary extensions to the whole cotangent) by the canonical symplectic form, two non-integrable brackets are obtained. The first decomposition is due to de Le\'on and Mart{\'\i}n de Diego \cite{deLeon1996} and the second one to Bates and Sniatycki \cite{Bates1993}. The advantage of this second decomposition is that it turns out to be symplectic, and it is the one we will use in the present paper. In any case, we proved that both brackets coincide on the submanifold of constraints \cite{Cantrijn1999a}. {Refer to Section 4 in \cite{Cantrijn1999a} for a detailed analysis explaining how this bracket generalizes the one defined by van der Schaft and Maschke in \cite{vanderSchaft1994}.}
	
	On the other hand, by studying the Hamilton--Jacobi equation, we develop a description of nonholonomic mechanics in the setting of skew-symmetric (or almost Lie) algebroids. Note that the ``almost'' is due to the lack of integrability of the distribution determining the constraints, showing the consistency of the description. In \cite{deLeon2010a,Grabowski2009} we defined a new almost Poisson bracket that we also called nonholonomic. So far, although both nonholonomic brackets have been used in these two different contexts as coinciding, no such proof has ever been published. This paper provides this evidence for the first time. 

But the issue does not end there. In 1951, R.~J.~Eden wrote his doctoral thesis on nonholonomic mechanics under the direction of P.A.M. Dirac ({\sl S-Matrix; Nonholonomic Systems}, University of Cambridge, 1951), and the results were collected in two publications \cite{Eden1951,Eden1951b}.
In the first paper, Eden introduced an intriguing $\gamma$ operator that mapped free states to constrained states. With that operator (a kind of tensor of type (1,1) that has the properties of a projector) 
Eden obtained the equations of motion, could calculate brackets of all observables, obtained a simple Hamilton--Jacobi equation, and even used it to construct a quantization of the nonholonomic system.
These two papers by Eden have had little impact despite their relevance. Firstly, because they were written in terms of coordinates that made their understanding difficult, and secondly, because it was not intil the 1980s when 
the study of nonholonomic systems became part of the mainstream of geometric mechanics. 

Recently, we have carefully studied these two papers by Eden, and realized that the operator $\gamma$ is nothing else a projection defined by the orthogonal
decomposition of the cotangent bundle provided by the Riemannian metric given by the kinetic energy. Consequently, we have defined a new bracket that we call Eden bracket, and 
proved that coincides with the previous nonholonomic brackets. 
 We are sure that this new approach to the dynamics of nonholonomic mechanical systems opens
new and relevant lines of research. Furthermore, this paper may be used as a reference for the reader interested in the different bracket formulations of nonholonomic mechanics.

The paper is structured as follows. In Section~\ref{sec:review_mechanics}, we review some elementary notions on Lagrangian and Hamiltonian mechanics within a geometric framework.
In Section~\ref{sec:nonholonomic}, we recall the main aspects of nonholonomic mechanics and present the corresponding dynamics in both Lagrangian and Hamiltonian settings. We also briefly discuss the skew-symmetric algebroid approach.
In Section~\ref{sec:nh_bracket}, we introduce the nonholonomic bracket as defined using the cotangent bundle approach and the symplectic decomposition of its tangent bundle along the constraint submanifold. Additionally, we define the nonholonomic bracket in the skew-symmetric algebroid context.
In Section~\ref{sec:Eden} we introduce the notion of Eden bracket.
The main results of the paper are presented in Section~\ref{sec:comparison}, where we prove that these three almost Poisson brackets coincide.
In Section~\ref{sec:HJ}, we show how the Eden approach is very useful to discuss Hamilton--Jacobi theory for nonholonomic mechanical systems. 
The above results are illustrated with two examples in Section~\ref{sec:example}: the nonholonomic particle and the rolling ball.
Finally, in Section~\ref{sec:conclusions}, we point out some interesting future lines of research opened up by the results of this paper.

We remark that some of the results in this paper (Theorems 1 and 5) were presented preliminarily, without proofs, in the form of a conference paper \cite{deLeon2023}.

\section{Lagrangian and Hamiltonian mechanics: a brief survey}\label{sec:review_mechanics}

\subsection{Lagrangian mechanics}\label{subsec:review_Lagrangian}

Let $L : TQ  \to \mathbb R$ be a Lagrangian function, 
where $Q$ is a configuration $n$-dimensional manifold. Then, $L = L(q^i, \dot{q}^i)$, where
$(q^i)$ are coordinates in $Q$ and $(q^i, \dot{q}^i)$ are the induced bundle coordinates in $TQ$.
We denote by $\tau_Q : TQ \to Q$ the canonical projection such that
$\tau_Q(q^i, \dot{q}^i) = (q^i)$.

We will assume that $L$ is regular, that is, the Hessian matrix
$$
\left( \frac{\partial^2 L}{\partial \dot{q}^i \partial \dot{q}^j} \right)
$$
is non--degenerate.
Using the canonical endomorphism $S$ on $TQ$ locally defined by
$$
S = d q^i \otimes \frac{\partial}{\partial \dot{q}^i}\, ,
$$
one can construct a 1-form $\theta_L$ defined by
$$
\theta_L = S^* (dL)\, ,
$$
and the 2-form
$$
\omega_L = - d\theta_L\, .
$$
Then, $\omega_L$ is symplectic if and only if $L$ is regular.

Consider now the vector bundle isomorphism
\begin{eqnarray*}
&&\flat_L  :  T(TQ) \to T^*(TQ)\\
&&\flat_L (v) = i_v \, \omega_L\, ,
\end{eqnarray*}
and the Hamiltonian vector field
$$
\xi_L = X_{E_L} \, ,
$$
defined by 
$$
\flat_L(\xi_L) = dE_L\, ,
$$
where $E_L = \Delta(L) -L$ is the energy, {and $\Delta$ is the Liouville vector field generating the dilations on the fibers of $TQ$. In bundle coordinates, $\Delta = \dot{q}^i {\partial}/{\partial \dot{q}^i}$}.
The vector field $\xi_L$, called the \emph{Euler--Lagrange vector field}, is locally given by
\begin{equation}
\xi_ L = \dot{q}^i \, \frac{\partial}{\partial q^i} + B^i \, \frac{\partial}{\partial \dot{q}^i}\, ,
\end{equation}
where
\begin{equation}
B^i \, \frac{\partial}{\partial \dot{q}^i}\left(\frac{\partial L}{\partial \dot{q}^j}\right) + \dot{q}^i \, 
\frac{\partial}{\partial q^i}\left(\frac{\partial L}{\partial \dot{q}^j}\right) - \frac{\partial L}{\partial q^j} = 0\,.
\end{equation}
Now, if $(q^i(t), \dot{q}^i(t))$ is an integral curve of $\xi_L$, then it satisfies the 
usual Euler--Lagrange equations
\begin{equation}\label{eleqs}
\dot q^i = \frac{dq^i}{dt}, \qquad  \frac{d}{dt} \left(\frac{\partial L}{\partial \dot{q}^i}\right) - \frac{\partial L}{\partial q^i} = 0\, .
\end{equation}

\subsection{Legendre transformation}\label{subsec:Legendre}

Let us recall that the Legendre transformation $FL : TQ
\to T^*Q$ is a fibred mapping (that is, $\pi_Q \circ FL
= \tau_Q$, where $\pi_Q : T^*Q
\to Q$ denotes the canonical projection of the cotangent bundle of $Q$). Indeed, $FL$ is the fiber derivative of $L$.

In local coordinates, the Legendre transformation is given by
$$
FL (q^i, \dot{q}^i) = (q^i, p_i), \quad p_i = \frac{\partial L}{\partial \dot{q}^i}\, .
$$
Hence, $L$ is regular if and only if $FL$ is a local diffeomorphism.

Along this paper we will assume that $FL$ is in fact a global
diffeomorphism (in other words, $L$ is hyperregular), which is the
case when $L\colon TQ\to \RR$ is a Lagrangian of mechanical type, namely
$$ L(v_q) = \frac{1}{2} \, g_q(v_q, v_q) - V(q)\, , $$
for $v_q \in T_qQ, q \in Q$,
where $g$ is a Riemannian metric on $Q$ and $V: Q \to \mathbb{R}$ is a potential energy.

\subsection{Hamiltonian description}\label{subsec:review_Hamiltonian}

The Hamiltonian counterpart is developed on the cotangent bundle
$T^*Q$ of $Q$. Denote by $\omega_Q = dq^i \wedge dp_i$ the canonical
symplectic form, where $(q^i, p_i)$ are the canonical coordinates on
$T^*Q$.

The Hamiltonian function is just $H = E_L \circ FL^{-1}$ and the
Hamiltonian vector field is the solution of the symplectic equation
$$
i_{X_H} \, \omega_Q = dH\, .
$$
The integral curves $(q^i(t), p_i(t))$ of $X_H$
satisfy the Hamilton equations
\begin{equation}\label{hamiltoneqs}
\begin{array}{lcr}
\dot{q}^i & = & \displaystyle{\frac{\partial H}{\partial p_i}}\, ,\\\\

\dot{p}_i &=& \displaystyle{- \frac{\partial H}{\partial q^i}}\, .
\end{array}
\end{equation}
Since $FL^* \omega_Q = \omega_L$, we deduce that $\xi_L$ and $X_H$
are $FL$-related, and consequently $FL$ transforms the solutions of the
Euler--Lagrange equations~\eqref{eleqs} into the solutions of the Hamilton equations~\eqref{hamiltoneqs}.

On the other hand, we can define a bracket of functions, called the canonical Poisson bracket, 
$$
\{ \, , \,\}_{can} : C^\infty(T^*Q) \times C^\infty (T^*Q) \to C^\infty(T^*Q)\, ,
$$
as follows
$$
\{ F , G \}_{can} = \omega_Q(X_ F, X_G) = X_G(F) = - X_F(G)\, .
$$

The local expression of the Poisson bracket is
$$
\{ F , G \}_{can} = \frac{\partial F}{\partial q^i} \frac{\partial G}{\partial p_i} - \frac{\partial F}{\partial p_i} \frac{\partial G}{\partial q^i}\, .
$$
	
\begin{remark}\label{localPbracket}
	If $(q^i)$ are local coordinates on $Q$, $\{e_i\} = \{e_i = e^j_i \partial/\partial q^j \}$ is a local basis of vector fields on $Q$, and
	{$\{\mu^i=\mu^i_j \dd q^j\}$ is the dual basis of 1-forms (that is, $\mu^i_j e^j_k = \delta^i_k$)}, then we van consider the corresponding local coordinates $(q^i, \pi_i)$ on $T^*Q$. {In fact, if $\alpha_q = \pi_i \mu^i(q)\in T_q^* Q$, then $(q^i, \pi_i)$ will be the local coordinates for $\alpha_q$. Under these considerations, we have that}
	\begin{equation}
	\label{Local-Poisson-bracket}
	\begin{aligned}
		& \{\pi_i, \pi_j\}_{can} = - C^k_{ij} (q) \pi_k\, , \\
		& \{q^i, \pi_j\}_{can} = e^i_j (q)\, , \\
		& \{q^i, q^j\}_{can} = 0\, ,
	\end{aligned}
	\end{equation}
	where
	$$
	[e_i, e_j ] = C^k_{ij} (q) e_k\, .
	$$
	Here $[\, , \, ]$ denotes the Lie bracket of vector fields (see \cite{deLeon2005}).
\end{remark}

The bracket $\{ \, , \,\}_{can} $ is a Poisson bracket, that is, $\{ \, , \,\}_{can} $ is $\mathbb{R}$-bilinear and:
\begin{itemize}
\item It is skew-symmetric: $\{ G , F \}_{can} = - \{ F , G \}_{can}$;

\item It satisfies the Leibniz rule:
$$
\{ F F' , G \}_{can} = F \{ F', G \}_{can} + F'\{ F , G \}_{can};
$$
and
\item It satisfies the Jacobi identity:
$$
\{ F , \{ G, H\}_{can} \}_{can} + \{ G , \{H, F\}_{can} \}_{can} + \{ H , \{F, G\}_{can} \}_{can} = 0 \, .
$$
\end{itemize}

	\bigskip

	Moreover, the Poisson bracket $\{ \, , \, \}_{can}$ may be used to give the evolution of an observable $F \in C^\infty(T^*Q)$,
	$$
	\dot{F}= X_H(F) = \{ F, H\}_{can} \; .
	$$
	
	\begin{remark}\label{dos}
	Given a 1-form $\alpha$ on a manifold $N$, we can define a vertical vector field $\alpha^V$ on
	the cotangent bundle using the formula
	\begin{equation}\label{lift}
	i_ {\alpha^V} \, \omega_N = - (\pi_N)^* \alpha \, .
	\end{equation}
	where $\omega_N$ is the canonical symplectic form  on $T^*N$ and $\pi_N : T^*N \to N$
	is the canonical projection.
	
	If $\alpha = \alpha_i dx^i$ is the local expression of $\alpha$ in coordinates $(x^i)$ on $N$, then
	$$
	\alpha^V =  \alpha_ i \, \frac{\partial}{\partial p_i}
	$$
	in bundle coordinates $(x^i, p_i)$ on $T^*N$. Thus, if $\beta_x \in T_x^*N$, we have that
	$$
	\alpha^V (\beta_x) = \frac{d}{dt}_{|_{t=0}} (\beta_x + t \alpha(x))\, .
	$$
	The vector field $\alpha^V$ is called the vertical lift of $\alpha$ to $T^*N$ (see \cite{deLeon1989a,Yano1973}).
	
	
	\end{remark}

\section{Nonholonomic mechanical systems} \label{sec:nonholonomic}

\subsection{The Lagrangian description} \label{subsec:nh_Lagrangian}
	
	A \emph{nonholonomic mechanical system} is a quadruple $(Q, g, V, D)$ where
	\begin{itemize}
	
	\item $Q$ is the configuration manifold of dimension $n$;
	
	\item $g$ is a Riemannian metric on $Q$;
	
	\item $V$ is a potential energy, $V \in C^\infty(Q)$;
	
	\item $D$ is a non-integrable distribution of rank $k < n$ on $Q$.
	\end{itemize}
	
	As in Subsection~\ref{subsec:Legendre}, the metric $g$ and the potential energy $V$ define a Lagrangian function $	L : TQ \to \mathbb{R}$ of mechanical type by
	$$
	L(v_q) = \frac{1}{2} \, g_q(v_q, v_q) - V(q)\, ,
	$$
	for $v_q \in T_qQ, q \in Q$.
	In bundle coordinates $(q^i, \dot{q}^i)$ we have
	$$
	L(q^i, \dot{q}^i) = \frac{1}{2} \, g_{ij}(q) \, \dot{q}^i  \dot{q}^j - V(q^i)\, .
	$$
	The nonholonomic dynamics is provided by the Lagrangian $L$ subject to the nonholonomic constraints given by $D$, which means that the permitted velocities should belong to $D$.
	The nonholonomic problem is to solve the equations of motion
    \begin{equation}\label{eq:nh_Lagrangian_coords}
    \begin{aligned}
        &\frac{d}{dt} \left(\frac{\partial L}{\partial \dot{q}^i}\right) - \frac{\partial L}{\partial q^i} = \lambda_A \mu^{A}_{ i}(q)\, ,\\
	    & \mu^{A}_{ i} (q) \dot{q}^i = 0\, ,
    \end{aligned}
    \end{equation}
	where $\{\mu^A\}$ is a local basis of $D^{\circ}$ (the annihilator of $D$) such that $\mu^A = \mu^{A}_{ i} \, dq^i$.
	Here $\lambda_A$ are Lagrange multipliers to be determined.
	
	A geometric description of the equations above can be obtained using the symplectic form $\omega_L$ and the vector bundle of 1-forms, $F$, defined by $F = \tau_Q^*(D^{\circ})$. More specifically, equations~\eqref{eq:nh_Lagrangian_coords} are equivalent to
	\begin{equation}
\begin{aligned}
	&& i_X \, \omega_L - dE_L \in \tau_Q^*(D^{\circ})\, , \\
	&& X \in TD\, .
	\end{aligned}
\end{equation}
	These equations have a unique solution, $\xi_{nh}$, which is called the \emph{nonholonomic vector field}.
	
	The Riemannian metric $g$ induces a linear isomorphism
	\begin{align}
	   \flat_g(q)\colon  T_q Q &\to T_q^*Q\\
	   v_q &\mapsto \flat_g(q)(v_q) = i_{v_q} g\, ,
	\end{align}
	and also a vector bundle isomorphism over $Q$
	$$
	\flat_g :TQ \to T^*Q\, ,
	$$
	and an isomorphism of $C^\infty(Q)$-modules
	$$
	\flat_g\colon \mathfrak{X}(Q) \ni X\mapsto  i_X g\in \Omega^{1}(Q)\, .
	$$
	The corresponding inverses of the three morphisms $\flat_g$ will be denoted by $\sharp_g$.
	
	We can define the orthogonal complement, $D^{\perp_g}$, of $D$ with respect to $g$, as follows:
	
	$$
	D_q^{\perp_g} = \{v_q \in T_qQ \; | \; g(v_q, w_q) = 0, \forall\,  w_q \in D \}\, .
	$$
	The set $D^{\perp_g}$ is again a distribution on $Q$, or, if we prefer, a vector sub-bundle of $TQ$ such that we have the Whitney sum
	
	\begin{equation}\label{dec1}
	TQ = D \oplus D^{\perp_g}\, .
	\end{equation}

\subsection{The Hamiltonian description} \label{subsec:nh_Hamiltonian}
	
	We can obtain the Hamiltonian description of the nonholonomic system $(Q, g, V, D)$ using the Legendre transformation
	$$
	FL : TQ \to T^*Q\, ,
	$$
	which in our case coincides with the isomorphism $\flat_g$ associated to the metric $g$.
	Indeed,
	$$
	FL (q^i, \dot{q}^i) = \left(q^i, \frac{\partial L}{\partial \dot{q}^i} \right) = (q^i, p_i=g_{ij}\dot{q}^j)\, .
	$$
	We can thus define the corresponding:
	\begin{itemize}
	
	\item Hamiltonian function $H= E_L \circ (FL)^{-1} : T^*Q \to \mathbb{R}$
	
	\item constraint submanifold 
	$M = FL(D) = \flat_g(D) = (D^{\perp_g})^{\circ}$.
	
	\end{itemize}
	
	Therefore, we obtain a new orthogonal decomposition (or Whitney sum)
	
	\begin{equation}\label{dec2}
	T^*Q = M \oplus D^{\circ}\, ,
	\end{equation}
	since 
	$$
	FL (D^\perp) = \flat_g(D^{\perp_g}) = D^{\circ}\, .
	$$
	This decomposition is orthogonal with respect to the induced metric on tangent covectors, and it is the
	translation of the decomposition \eqref{dec1} to the Hamiltonian side.
	Similarly to the Lagrangian framework, $M$ and $D^{\circ}$ are vector sub-bundles of $\pi_Q : T^*Q \to Q$ over $Q$.
	We have the canonical inclusion
	\begin{equation}
		i_M : M \hookrightarrow T^*Q\, ,
	\end{equation} 
	and the orthogonal projection
	\begin{equation}\label{eq:gamma}
		\gamma : T^*Q \to  M\, .
	\end{equation}

The equations of motion for the nonholonomic system on $T^*Q$ can
now be written as 
\begin{equation}\label{eq:Hamilton_nh_coords}
\begin{aligned}
&\dot q^i=\displaystyle{\frac{\partial H}{\partial p_i}} \vspace{0.3cm}\, ,\\
&\dot p_i=\displaystyle{-\frac{\partial
H}{\partial q^i}- \bar{\lambda}_A \mu^{A}_{ j}g^{ij}}\, .
\end{aligned}
\end{equation}
together with the constraint equations
$$
\mu^{A}_{ i} g^{ij}p_j = 0\, .
$$
Notice that here the $\bar{\lambda}_A$'s are Lagrange multipliers to be determined.

Now the vector bundle of constrained forces generated by the 1-forms $\tau_Q^*(\mu^A)$,
can be translated to the cotangent side and we obtain the vector bundle generated by the 1-forms
$\pi_Q^*(\mu^A)$, say $\pi_Q^*(D^{\circ})$. Therefore, 
the nonholonomic Hamilton equations~\eqref{eq:Hamilton_nh_coords} can
be rewritten in intrinsic form as
\begin{equation}\label{eq:Int-Def-X-nh}
\begin{array}{rcl}
(i_X\omega_Q-dH)_{|M}&\in& \pi_Q^*(D^{\circ})\, , \\
X_{|M} &\in& TM\, .
\end{array}
\end{equation}
These equations have a unique solution, $X_{nh}$, which is called the \emph{nonholonomic vector field}. The vector fields $X_{nh}$ and $\xi_{nh}$ are related by the Legendre transformation restricted to $D$, namely,
\begin{equation} \label{eq:Xnh_Legendre}
	T(FL)_{|D}(\xi_{nh}) = X_{nh} \circ (FL)_{|D} \, .
\end{equation}
	
\subsection{The skew-symmetric algebroid approach} \label{subsec:algebroid}
	
	In \cite{deLeon2010a} (see also \cite{Grabowski2009}) we have developed an approach to nonholonomic mechanics based on the skew-symmetric algebroid setting.
	
	We denote by $i_D : D \to TQ$ the canonical
	inclusion. The canonical projection given by the decomposition $TQ = D \oplus D^\perp$ on $D$
	is denoted by $P : TQ \to D$.

	Then, the vector bundle $(\tau_Q)_{|_D} : D \to Q$ is an skew-symmetric algebroid.
	The anchor map is just the canonical inclusion $i_D : D \to TQ$, and the skew-symmetric bracket $\|\, , \, \|$ on the space of sections $\Gamma(D)$ is given by
	$$
	\|X, Y\| = P([X, Y])\, ,
	$$
	for $X, Y \in \Gamma(D)$.
	Here, $[\, , \, ]$ is the standard Lie bracket of vector fields.
	
	We also have the vector bundle morphisms provided by the adjoint operators:
	
	\begin{eqnarray*}
	 && i_{D}^* : T^*Q \to D^*\, ,\\
	 && P^* : D^* \hookrightarrow  T^*Q\, ,
	\end{eqnarray*}
	where  $D^*$ is the dual vector bundle of $D$.
	
	We define now an almost Poisson bracket on $M$ as follows (see \cite{deLeon2010a}):
	\begin{equation}\label{bracket-D-star}
	\begin{aligned}
		&\{\,  , \, \}_{D^*} : C^\infty(D^*) \times  C^\infty(D^*) \to C^\infty(D^*)\, ,\\
		&\{\phi , \psi\}_{D^*} = \{\phi \circ i_D^*, \psi \circ i_D^*\}_{can} \circ P^*\, .
	\end{aligned}
	\end{equation}
	
	  \begin{remark}\label{remark:PbracketDstar}
        Suppose that $(q^i)$ are local coordinates on $Q$, and that $\{e_i\} = \{e_a, e_A\}$ is a local basis of vector fields on $Q$ such that 
        $\{e_a\}$ (resp.~$\{e_A\}$) is a local basis of $\Gamma(D)$ (resp.~$\Gamma(D^{\perp_g})$) with
        $$
        e_i = e^j_i \frac{\partial}{\partial q^j}\, .
        $$
        Then, we can consider the dual basis $\{\mu^i\} = \{ \mu^a, \mu^A\}$ of 1-forms on $Q$ and the corresponding local coordinates
        $(q^i, v^i) = (q^i, v^a, v^A)$ on $TQ$ and
        $(q^i, p_i) = (q^i, \pi_a, \pi_A)$ on $T^*Q$. It is clear that
        $(q^i, v^a)$ (resp.~$(q^i, y_a = \pi_a \circ P^*$)) are local coordinates on $D$ (resp.~on $D^*$). In addition,
        we have the following simple expressions of 
        $i_D : D \to TQ$, $P : TQ \to D$ and their dual morphisms 
        
	\begin{equation}
	\begin{aligned}\label{eq:localexpression-morphisms}
	i_D &:&(q^i, v^a) \mapsto (q^i, v^a, 0)\, , \\
	i_D^*  &:& (q^i, \pi_a, \pi_A) \mapsto (q^i, \pi_a)\, , \\
	P &:& (q^i, v^a, v^A) \mapsto (q^i, v^a) \, , \\
        P^* &:& (q^i, y_a) \mapsto (q^i, y_a, 0) \, .
	\end{aligned}
	\end{equation}
    Hence, using equations~\eqref{Local-Poisson-bracket} and \eqref{bracket-D-star} we deduce that
	\begin{equation}\label{localexpression-bracket-D-star}
	\begin{aligned}
		& \{y_a, y_b\}_{D^*} = - C^c_{ab}(q) y_c \, ,\\
		& \{q^i, y_a \}_{D^*} = e^i_a(q)\, , \\
		& \{q^i, q^j\}_{D^*} = 0\, ,
	\end{aligned}
	\end{equation}
       where
       $$
       [e_a, e_b] = C^c_{ab} (q) e_c + C^A_{ab} (q) e_A\, .
       $$
        \end{remark}

     The bracket $\{\, , \, \}_{D^*}$ has the same properties as a Poisson bracket, although it may not satisfy the Jacobi identity, that is, $\{\, , \, \}_{D^*}$ is $\mathbb{R}$-bilinear and
         \begin{itemize}
         \item It is skew-symmetric : $\{\psi , \phi\}_{D^*} = - \{\phi, \psi\}_{D^*}$;
    
         \item It satisfies the Leibniz rule in each argument:
         $$
         \{\phi \phi' , \psi\}_{D^*} = \phi \{\phi' , \psi \}_{D^*} + \phi' \{\phi , \psi \}_{D^*}
         $$
         \end{itemize}
         
         However, one may prove that $\{\, , \, \}_{D^*}$ is a Poisson bracket if and only if the distribution $D$ is integrable (see \cite{Garcia-Naranjo2020}).

        Moreover, if
        $$
        FL_{nh} : D \to D^*
        $$
        is the nonholonomic Legendre transformation given by
        $$
        FL_{nh} = i_{D^*} \circ FL \circ i_D\, ,
        $$
        and $Y_{nh}$ is the nonholonomic dynamics in $D^*$,
        $$
        T(FL_{nh}) (\xi_{nh}) = Y_ {nh} \circ FL_{nh} \; ,
        $$ 
        then the bracket  $\{\, , \,\}_{D^*}$ may be used to give the evolution of an observable $\phi \in C^\infty(D^*)$. In fact, if 
        $h : D^* \to \mathbb{R}$ is the constrained Hamiltonian function defined by
        $$
        h = (E_L)_{|D} \circ FL_{nh}^{-1}\, ,
        $$
        we have that
        $$
        \dot{\phi} = Y_{nh}(\phi) = \{\phi, h\}_{D^*}\, , 
        $$
        for each $\phi \in C^\infty(D^*)$.

\section{The nonholonomic bracket}\label{sec:nh_bracket}

	Consider the vector sub-bundle $T^DM$ over $M$ defined by
	$$
	T^DM = \{Z \in TM \; | \; T\pi_Q(Z) \in D \}\, .
	$$
	As we know \cite{Bates1993,Cantrijn1999a}, $T^DM$ is a symplectic vector sub-bundle of the symplectic vector bundle 
	$(T_M(T^*Q), \omega_Q)$, where the restriction of $\omega_Q$ to any fiber of $T_M(T^*Q)$ is also denoted by $\omega_Q$.
	Thus, we have the following symplectic decomposition
	\begin{equation}\label{sym2}
	T_M(T^*Q) = T^DM \oplus (T^DM)^{\perp_{\omega_Q}}\, ,
	\end{equation}
	where $(T^DM)^{\perp_{\omega_Q}}$ denotes the symplectic orthogonal complement of $T^DM$.
	Therefore, we have associated projections
	\begin{equation}\label{eq:projs_P_Q}
	\begin{aligned}
		{\cal P} : T_M(T^*Q) & =  T^DM  \oplus (T^DM)^{\perp_{\omega_Q}} \to T^DM \, ,\\
		{\cal Q} : T_M(T^*Q) & =  T^DM  \oplus (T^DM)^{\perp_{\omega_Q}} \to (T^DM)^{\perp_{\omega_Q}}\, .
	\end{aligned}
	\end{equation}
One of the most relevant applications of the above decomposition is that
$$
X_{nh} = {\cal P} (X_H)
$$
along $M$, {where $X_{nh}$ denotes the nonholonomic vector field defined by \eqref{eq:Int-Def-X-nh}}.

In addition, the above decomposition allows us to define the so-called nonholonomic bracket as follows.
Given $f, g \in C^\infty(M)$, we set

\begin{equation}\label{eq:nhb}
\{f, g\}_{nh} = \omega_Q({\cal P}(X_{\tilde{f}}), {\cal P}(X_{\tilde{g}})) \circ i_M\, ,
\end{equation}
where $i_M : M \to T^*Q$ is the canonical inclusion, and $\tilde{f}, \tilde{g}$ are arbitrary extensions to $T^*Q$ of $f$ and $g$, respectively (see \cite{Cantrijn1999a,Ibort1999}).
Since the decomposition \eqref{sym2} is symplectic, one can equivalently write

\begin{equation}\label{nhb2}
\{f, g\}_{nh} = \omega_Q(X_{\tilde{f}}, {\cal P}(X_{\tilde{g}}) ) \circ i_M \; .
\end{equation}

\begin{remark}
Notice that $f\circ \gamma$ and $g\circ \gamma$ are natural extensions of $f$ and $g$ to $T^*Q$. Hence, we can also define the above nonholonomic bracket as follows
\begin{equation}\label{nhb3}
\{f, g\}_{nh} = \omega_Q(X_{f \circ \gamma}, {\cal P}(X_{g \circ \gamma}) ) \circ i_M \; .
\end{equation}
\end{remark}

The bracket $\{\,, \,\}_{nh}$ is an almost Poisson bracket on $M$. In fact, $\{\,, \,\}_{nh}$ satisfies the Jacobi identity if and only if the distribution $D$ is integrable (see \cite{Ibort1999,vanderSchaft1994}). In addition, 
if $H_M : M \to \mathbb{R}$ is the constrained Hamiltonian function on $M$, namely
$$
H_M = H \circ i_M\, ,
$$
then, using the nonholonomic bracket, we can obtain the evolution of an observable $f \in C^\infty(M)$ as follows
\begin{equation}\label{prima}
\dot{f} = X_{nh}(f) = \{f, H_M\}_{nh} \; .
\end{equation}

\begin{remark}\label{nuevo}
If $x \in M$ and $f \in C^\infty(M)$ then, using equations~\eqref{nhb3} and \eqref{prima}, we deduce that

$$
(X_{nh}(x))(f) = ({\cal P}(X_{H_M \circ \gamma})(x))(f \circ \gamma)\, ,
$$
but as ${\cal P}(X_{H_M \circ \gamma})(x) \in T_xM$ and $f \circ \gamma$ is an extension of $f$ to $T^*Q$, it follows that
$$
X_{nh}(x) = {\cal P}(X_{H_M \circ \gamma})(x)\, .
$$
\end{remark}

{See \cite[Section 4]{Cantrijn1999a} for a thorough comparison between this bracket and the one defined by van der Schaft and Maschke in \cite{vanderSchaft1994}.}

\section{Eden bracket} \label{sec:Eden}
	
	{Using the orthogonal projector $\gamma : T^*Q \to M$ as defined in \eqref{eq:gamma}}, we can define another almost Poisson bracket on $M$ as follows:
	\begin{equation}
	\begin{aligned}\label{eq:Eden_bracket}
	&\{\, , \, \}_E : C^\infty(M) \times  C^\infty(M) \to C^\infty(M)\\
	&\{f , g\}_E = \{f \circ \gamma, g \circ \gamma\}_{can} \circ i_M\, .
	\end{aligned}
	\end{equation}
	
	This bracket will be called \emph{Eden bracket}.
	
	\begin{remark}
	Let $(q^i, \pi_a, \pi_A)$ be local coordinates on $T^*Q$ as in Remark \ref{remark:PbracketDstar}. Then, we have that
	the constrained submanifold $M = (D^{\perp_g})^{\circ}$ is locally described by
	$$
	M = \{ (q^i, \pi_a, \pi_A) \in T^*Q \; | \; \pi_A = 0 \}\, .
	$$
	Thus, $(q^i, \pi_a)$ are local coordinates on $M$ and the expression of the inclusion $i_ M : M \to T^*Q$ is
	$$
	i_M (q^i, \pi_a) = (q^i, \pi_a, 0)\, .
	$$
	Hence, using equations~\eqref{Local-Poisson-bracket} and \eqref{eq:Eden_bracket}, we deduce that Eden bracket is locally characterized by
	\begin{equation}\label{localexpression-E-bracket}
	\begin{aligned}
		& \{\pi_a, \pi_b\}_E = - C^c_{ab}(q) \pi_c\, ,\\
		& \{q^i, \pi_a\}_E = e^i_a (q)\, ,\\
		& \{q^i, q^j\}_E = 0\, .
	\end{aligned}
	\end{equation}
	\end{remark}
	As the bracket $\{\, , \, \}_{D^*}$, the Eden bracket satisfies all the properties
	of a Poisson bracket, with the possible exception of the Jacobi identity.

\section{Comparison of brackets} \label{sec:comparison}
	
	First of all, we shall prove that the almost Poisson brackets defined on $D^*$ and $M$ are isomorphic.

		\begin{theorem}\label{th1}
	The vector bundle isomorphism
	$$
	i_{M, D^*} : M \to D^*
	$$
	over the identity of $Q$, given by the composition
	$$
	i_{M, D^*}  = i_D^* \circ i_M\, ,
	$$
	is an almost Poisson isomorphism between the almost Poisson manifolds
	$(M, \{\, , \, \}_{E})$ and  $(D^*, \{\, , \, \}_{D^*})$ .

	\end{theorem}
	
	\begin{proof}
	Using that $M = (D^{\perp_g})^{\circ}$, it is easy to deduce that
	$i_{M, D^*}$ is an isomorphism of vector bundles over the identity of $Q$. Thus, it remains to be seen that
	$$
	\{\phi \circ i_{M, D^*}, \psi \circ i_{M, D^*}\}_E = \{\phi, \psi\}_{D^*} \circ i_{M, D^*}\, ,
	$$
	for all $\phi, \psi \in C^\infty(D^*)$.

	A direct proof comes from the commutativity of the following diagram:
	
	\begin{center}

\begin{tikzcd}
    & T^\ast Q \arrow[ld, "\gamma"'] \arrow[rd, "i_{D^\ast}"] &                                       \\
M \arrow[rr, "i_{M, D^\ast}"] \arrow[rd, "i_M"'] &                                                                   & D^\ast \arrow[ld, "P^\ast"] \\
    & T^\ast Q                                                          &                                      
\end{tikzcd}
\end{center}
{where $\gamma$ is the orthogonal projector induced by the Riemannian metric, as defined in \eqref{eq:gamma}.}
	In fact, given $\phi, \psi \in C^\infty(D^*)$, using equations~\eqref{bracket-D-star} and \eqref{eq:Eden_bracket}, and the following facts
	$$
	i_{M, D^*} \circ \gamma = i_{D^*} \; , \; i_M = P^* \circ i_{M, D^*}  \; ,
	$$ 
	we have
	
	\begin{align}
		\{\phi \circ i_{M, D^*}, \psi \circ i_{M, D^*}\}_E & = \{\phi \circ i_{M, D^*} \circ \gamma, \psi \circ i_{M, D^*} \circ \gamma\}_{can} \circ i_M  \\
		& = \{\phi \circ i_{D}^*, \psi \circ i_{D}^*\}_{can} \circ P^* \circ i_{M, D^*} \\
		& = \{\phi, \psi\}_{D^*} \circ i_{M, D^*}\, .
	\end{align}
	
	\end{proof}
	
	\begin{remark}\label{cinco}
	An alternative proof can be given if we consider adapted bases on $D$ and $D^{\perp_g}$.
	Indeed, the local basis $\{e_i\} = \{e_a, e_A\}$
	of vector fields on $Q$ such that $\{e_a\}$ is a local basis of $D$ and $\{e_A\}$ is a local basis of $D^{\perp_g}$, defines coordinates
	$(q^i, v^i) = (q^i, v^a, v^A)$ on $TQ$ as in Remark \ref{remark:PbracketDstar}.
	Therefore, $(q^i, v^a)$ and $(q^i, v^A)$ define coordinates on $D$ and $D^{\perp_g}$, respectively.
	
	Analogously, we can consider the dual local basis $\{\mu^i\} = \{\mu^a, \mu^A\}$, and the induced coordinates on $D^*$ and $T^*Q$, say $(q^i, y_a)$ and $(q^i, \pi_a, \pi_A)$, respectively, as in Remark \ref{remark:PbracketDstar}.
	
	In addition, $(q^i, \pi_a)$ are local coordinates on $M$ in such a way that the canonical inclusion 
	$i_M : M \to T^*Q$ is given by
	\begin{equation}\label{i-M}
	i_M (q^i, \pi_a) = (q^i, \pi_a, 0)\, .
	\end{equation}
	Therefore, using equations~\eqref{eq:localexpression-morphisms} and \eqref{i-M}, we obtain that the local expression of $i_{M, D^*} : M \to D^*$ is just the identity
	$$
	i_{M, D^*}  : (q^i, \pi_A) \to (q^i, \pi_A)\, ,
	$$
	and Theorem~\ref{th1} immediately follows from equations~\eqref{localexpression-bracket-D-star} and \eqref{localexpression-E-bracket}.
	
	
	\end{remark}
	
	Next, we will prove that the Eden bracket is just the nonholonomic bracket defined in \cite{Cantrijn1999a,Ibort1999}.
	
	\begin{proposition}\label{prop:2}
	{Given the projection ${\cal P}$ defined in \eqref{eq:projs_P_Q}}, we have that
	$$
	{\cal P} (Z) = T\gamma(Z)\, ,
	$$
	for every $Z \in T^D(T^*Q) = \{ Y \in T(T^*Q) \; | \; (T\pi_Q)(Y) \in D\}$,
	{where $\gamma$ is the projection defined in \eqref{eq:gamma}.}
	\end{proposition}
	
	\begin{proof}
	Suppose that $Z \in T_{\beta}^DM$, with $\beta \in T^*_qQ$. Thus, $T\pi_Q(Z) \in D$. Then, we have
	$$
	T\pi_Q(T\gamma(Z)) = T(\pi_Q \circ \gamma)(Z) = T\pi_Q(Z) \in D\, ,
	$$
	which, using that $T\gamma$ takes values in $TM$, implies that $T\gamma(Z) \in T^DM$.
	
	Next, we will prove that 
	$$
	Z-T\gamma (Z) = (\epsilon_q)^V_\beta\, ,
	$$ 
	with $\epsilon_q \in D^{\circ}_q$, 
        where $(\epsilon_q)^V_{\beta_q} \in T_{\beta_q}(T^*Q)$ is just the vertical lift of $\epsilon_q$ to $T_{\beta_q} (T^*Q)$ defined by
	\begin{equation}\label{eq:vertical-lift-point}
		(\epsilon_q)^V_{\beta_q} = \frac{d}{dt}(\beta_q + t \epsilon_q) 
	\end{equation}
	(see Remark \ref{dos}).
	
	Indeed,
	$$
	T\pi_Q(Z-T\gamma(Z)) = T\pi_Q(Z) - T\pi_Q(T\gamma(Z)) = T\pi_Q(Z) - T\pi_Q(Z) = 0\, ,
	$$
	then $Z - T\gamma(Z)$ is a vertical tangent vector, and hence $Z - T\gamma(Z) = (\epsilon_q)^V_{\beta_q}$, for some 1-form $\epsilon_q \in T_q^*Q$, with $q \in Q$.
	
	Let $X : Q \to D$ be a section of the vector sub-bundle $D$, and denote by $\hat{X}$ its associated fiberwise linear function:
	\begin{eqnarray*}
	&&\hat{X} : T^*Q \to \mathbb{R}\, ,\\
	&&\hat{X} (q^i, p_i) = X^ip_i\, ,
	\end{eqnarray*}
	where $X = X^i {\partial}/{\partial q^i}$.
	
	Then, we have
	$$
	(Z - T\gamma(Z)) (\hat{X}) = Z(\hat{X}) - Z(\hat{X} \circ \gamma)\, ,
	$$
	but
	$$
	(\hat{X} \circ \gamma) (q^i, p_i) = \hat{X}(q^i, \gamma^k_ip_k) = X^i \gamma^k_ip_k = \hat{X}(q^i, p_i)\, ,
	$$
	since we are assuming that we are taking tangent vectors at a point $(q^i, p_i) \in M$, which implies that $\gamma^k_ip_k = p_i$.
	
	Therefore,
	$$
	(Z - T\gamma(Z)) (\hat{X}) = 0\, ,
	$$
	or, equivalently,
	$$
	(\epsilon_q)^V_{\beta_q}(\hat{X}) = \epsilon_q(X(q)) = 0\, ,
	$$
	for all $X \in \Gamma (D)$.
	This proves that $\epsilon_q \in D^{\circ}_q$.
	
	Now, we will see that
	\begin{equation}\label{OrtogonalDM}
	(\epsilon_q)^V_{\beta_q} \in (T_{\beta_q}^DM)^{\perp_{\omega_Q}}\,.
	\end{equation}
	Indeed, if $W \in T^D_{\beta_q}M$, then, using standard properties of the canonical symplectic structure $\omega_Q$) (see equation~\eqref{lift}), 
	and the fact that $(T_{\beta_q}\pi_Q)(W) \in D$ and $\epsilon_q \in D^{\circ}$, we deduce that
	$$
	\omega_Q (\beta) ((\epsilon_q)^V_{\beta_q}, W) = - ((T_{\beta_q}^* \pi_Q)(\epsilon_q))(W) = -\epsilon_q ((T_{\beta_q} \pi_Q)(W)) = 0\, .
	$$
	This proves equation~\eqref{OrtogonalDM} and, thus,
	$$
	{\cal P}(Z - T\gamma(Z)) = 0\, .
	$$
	Since $T\gamma(Z) \in T^DM$, we have that ${\cal P}(T\gamma (Z)) = T\gamma(Z))$, which implies that
	$$
	{\cal P}(Z) = T\gamma(Z)\, .
	$$
	\end{proof}
	
	\begin{proposition}\label{prop:3}
	For any function $f \in C^{\infty}(M)$ and $x \in M$, we have
	$$
	(T_x\pi_Q)(X_{f \circ \gamma} (x)) \in D_{q}\, ,
	$$
	with $q = \pi_Q(x)$, {where $\gamma$ is the projection defined in \eqref{eq:gamma}}. In consequence,
	$$
	X_{f \circ \gamma} (x) \in T_x^D(T^*Q) \, ,
	$$
	for any $x\in M$.
	\end{proposition}
	
	\begin{proof}
	Let $\epsilon$ be a section of the vector bundle $ D^{\circ} \to Q$. Then, we have
	
	\begin{align}
		\langle \epsilon(q), (T_x\pi_Q)(X_{f \circ \gamma})(x) \rangle & = \langle (T_x^*\pi_Q) (\epsilon(q)), X_{f \circ \gamma}(x) \rangle \\
		& = - \omega_Q ((\epsilon_q)^V_x, X_{f \circ \gamma})(x) = \langle d(f \circ \gamma), \epsilon^V \rangle (x) \\
		& = \epsilon^V (x) (f \circ \gamma) = \frac{d}{dt}_{|_{t=0}} ((f \circ \gamma) (x + t \epsilon(q)))\\
		& = \frac{d}{dt}_{|_{t=0}}  f (\gamma (x) + t \gamma(\epsilon(q))) = 
		\frac{d}{dt}_{|_{t=0}}(f(\gamma(x)) \\
		& = 0\, ,
	\end{align}
         since $\gamma(\epsilon(q)) =0$.
	\end{proof}
	
		Using Remark \ref{nuevo} and Propositions~\ref{prop:2} and \ref{prop:3}, we conclude that
	\begin{corollary}\label{coro}
    For every $x\in M$, we have 
	$$
	X_{nh}(x) = (T_x\gamma)(X_{{H_M} \circ \gamma}(x))\, .
	$$
	\end{corollary}
	
	This result shows that $T\gamma$ does not project the Hamiltonian dynamics $X_H$ onto the nonholonomic dynamics
	$X_{nh}$.
    However, we can achieve this by modifying the Hamiltonian function using the projector $\gamma$, i.e.~considering considering $X_{H \circ \gamma}$ instead of $X_H$.

	\begin{theorem}
	{The nonholonomic bracket
	$\{\, , \, \}_{nh}$ defined by \eqref{eq:nhb} and the Eden bracket  $\{\, , \, \}_{E}$ defined by \eqref{eq:Eden_bracket} coincide on $M$, namely,
	\begin{equation}
		\{f, g\}_{nh}(x) = \{f, g\}_E(x)\, ,
	\end{equation}
	for any pair of functions $f, g\in C^\infty(M)$ and any point $x\in M$.}
	\end{theorem}

	\begin{proof}
	
	
	Indeed, if $x \in M$ then, using Propositions~\ref{prop:2} and \ref{prop:3}, we have
	\begin{eqnarray*}
	\{f, g\}_{nh}(x) &=& \omega_Q (X_{f \circ \gamma}, {\cal P}(X_{g \circ \gamma}) (x) \\
	&=& \omega_Q (X_{f \circ \gamma}, T\gamma (X_{g \circ \gamma}))(x)\\
	&=& d(f \circ \gamma) (x) (T\gamma(X_{g \circ \gamma})(x))\\
	&=& X_{g \circ \gamma}(x) (f \circ \gamma) \\
	&=& \{f \circ \gamma, g \circ \gamma\}_{can} (x) \\
	&=& \{f, g\}_E (x)\, .
	\end{eqnarray*}
		
	\end{proof}
		
	\begin{remark}
	In his paper, Eden writes the dynamics in terms of the constrained variables that he denotes by 
	$(q^{i*},p_i^*) = (q^i \circ \gamma, p_i \circ \gamma)$. 
	Then, he computes the Poison brackets of the observables substituting the canonical variables $(q^i,p_i)$ by the constrained variables $(q^{i*},p_i^*)$. 
	Indeed, this coincides with computing the Eden brackets of the original observables. This can be seen explicitly in equation (3.4) in \cite{Eden1951b}, 
	where Eden computes the commutation relations of the constrained variables. Indeed, if those are taken as structure constants, they define the Eden bracket.
	
	\end{remark}

\section{Application to the Hamilton--Jacobi theory}\label{sec:HJ}
	
\subsection{Hamilton--Jacobi theory for standard Hamiltonian systems}	
	
Given a Hamiltonian $H = H(q^i, p_i)$, the standard formulation of the Hamilton--Jacobi problem is to find a
function $S(t, q^i)$, called the \emph{principal function}, such that
\begin{equation}
\label{hj1}
\frac{\partial S}{\partial t} + H\left(q^i, \frac{\partial S}{\partial
q^i}\right) = 0\, .
\end{equation}
If we put $S(t, q^i) = W(q^i) - t E$, where $E$ is a constant, then
$W$ satisfies
\begin{equation}
\label{hj2}
H\left(q^i, \frac{\partial W}{\partial q^i}\right) = E\, .
\end{equation}
The function $W$ is called the \emph{characteristic function}.
Equations \eqref{hj1} and \eqref{hj2} are indistinctly referred as
the \emph{Hamilton--Jacobi equation}. See \cite{Abraham2008,Arnold1978} for more details.

Let $Q$ be the configuration manifold, and $T^*Q$ its cotangent
bundle equipped with the canonical symplectic form $\omega_Q$.
Let $H : T^*Q \to \mathbb{R}$ be a Hamiltonian function and $X_H$
the corresponding Hamiltonian vector field (see Subsection~\ref{subsec:review_Hamiltonian}).



Let $\lambda$ be a closed 1-form on $Q$, i.e.~$d\lambda=0$ (then,
locally $\lambda = dW$). We have that

\begin{theorem}\label{tseis}

The following conditions are equivalent:

\begin{enumerate}
\item[(i)] If $\sigma: I\to Q$ satisfies the equation
$$
\frac{dq^i}{dt} = \frac{\partial H}{\partial p_i}\, ,
$$
then $\lambda\circ \sigma$ is a solution of the Hamilton equations;

\item[(ii)] $d (H\circ \lambda)=0$.
\end{enumerate}
\end{theorem}

If $\lambda$ is a closed 1-form on $Q$, one may define a vector field on $Q$:
\begin{equation}\label{eq:X-H-lambda-1}
	X_H^{\lambda}=T\pi_Q \circ X_H\circ \lambda\, .
\end{equation}

The following conditions are equivalent:
\begin{enumerate}
\item[(i)] If $\sigma: I\to Q$ satisfies the equation
$$
\frac{dq^i}{dt} = \frac{\partial H}{\partial p_i}
$$
then $\lambda\circ \sigma$ is a solution of the Hamilton equations;
\item[(i)'] If $\sigma: I\to Q$ is an integral curve of
$X_H^{\lambda}$, then $\lambda\circ \sigma$ is an integral curve of
$X_H$;
\item[(i)''] $X_H$ and $X_H^{\lambda}$ are $\lambda$-related,
i.e.
$$
T\lambda(X_H^{\lambda})=X_H \circ \lambda\, .
$$
\end{enumerate}

Moreover, Theorem \ref{tseis} may be reformulated as follows.

\begin{theorem}\label{thm:HJ}
Let $\lambda$ be
a closed $1$-form on $Q$. Then the following conditions are
equivalent:

\begin{enumerate}
\item[(i)] $X_H^\lambda$ and $X_H$ are $\lambda$-related;

\item[(ii)] $d (H\circ \lambda)=0$\, .
\end{enumerate}

In that case, $\lambda$ is called a \emph{solution of the Hamilton--Jacobi problem} for $H$.
\end{theorem}

If $\lambda = \lambda_i(q) \, dq^i$, 
then $\lambda$ is a solution of the Hamilton--Jacobi problem if and only if
$$
H(q^i, \lambda_i (q^j)) = E\, ,
$$
for some constant $E$,
and we recover the classical Hamilton--Jacobi equation~\eqref{hj2} when
$$
\lambda_i = \frac{\partial W}{\partial q^i}\, .
$$


\begin{remark}\label{remark:mechanical_Hamiltonian}
	Suppose that the Hamiltonian function $H\colon T^\ast Q \to \RR$ is of mechanical type, that is,
	\begin{equation}\label{eq:H-mechanical}
		H(\alpha_q) = \frac{1}{2} g_q^\ast (\alpha_q, \alpha_q) + V(q)\, , 
	\end{equation}
	for $\alpha_q\in T^\ast_q Q$, with $V\in C^\infty(Q)$ and $g_q^\ast$ the scalar product on $T_q^\ast Q$ induced by the Riemannian metric $g$ on $Q$. Then, if $\lambda\in \Omega^1(Q)$ and $f\in C^\infty(Q)$, using equation~\eqref{eq:X-H-lambda-1}, we have that
	\begin{align}
		\langle df(q), X_H^\lambda(q) \rangle &= \langle d(f\circ \pi_Q) (\lambda_q), X_H(\lambda(q)) \rangle \\
		&= - \left(i_{(df(q))_{\lambda(q)}^V} \omega_Q(\lambda(q)) \right) \left(X_H(\lambda(q)\right) 
		= \left(i_{X_H}\omega_Q\right) (\lambda(q)) (df(q))_{\lambda(q)}^V \, ,
	\end{align}
	so, from equations~\eqref{eq:vertical-lift-point} and \eqref{eq:H-mechanical}, we deduce that
	\begin{equation}\label{eq:key-point}
		\langle df(q), X_H^\lambda(q)\rangle = \restr{\frac{d}{dt}}{t=0} H \left(\lambda(q) + t df(q)\right)
		= g_q^\ast (\lambda(q), df(q)) = \langle df(q), \sharp_g(\lambda(q)) \rangle\, .
	\end{equation}
	This implies that 
	\begin{equation}\label{eq:X-H-lambda-2}
		X_H^\lambda(q) = \sharp_g (\lambda(q))\, ,
	\end{equation}
	for any $q\in Q$. 
\end{remark}

One may find in the literature (see Theorem 2 in \cite{Carinena2006}) an extension of Theorem~\ref{thm:HJ} for the more general case in which the 1-form $\lambda$ is not necessarily closed.

\begin{theorem}\label{thm:HJ-generalized}
	Let $\lambda$ be a 1-form on $Q$. Then, the following conditions are equivalent:
	\begin{enumerate}
		\item $X_H^\lambda$ and $X_H$ are $\lambda$-related,
		\item $d(H\circ \lambda) + i_{X_H^\lambda} \dd \lambda = 0$.
	\end{enumerate}
\end{theorem}

The equation 
\begin{equation}\label{eq:generalized-HJ}
	d(H\circ \lambda) + i_{X_H^\lambda} \dd \lambda = 0
\end{equation}
will be called \emph{generalized Hamilton--Jacobi equation} for $H\colon T^\ast Q \to \RR$. 

In Subsection~\ref{sec:HJ_algebroid} (see Theorem~\ref{thm:HJ-generalized-nonholonomic}), we will prove a nonholonomic version of Theorem~\ref{thm:HJ-generalized}, which will be useful for our interests.


\subsection{Hamilton--Jacobi theory for nonholonomic mechanical  systems}

Let $H:T^*Q \to \mathbb{R}$ be a mechanical Hamiltonian function
subject to nonholonomic constraints given by a distribution $D$ on $Q$, as in the previous sections. We will continue using the same notations. Hence, 
we have the decomposition
$$
T^*Q = M \oplus D^{\circ} \, .
$$
The vector field $X_{nh} \in \mathfrak{X}(M)$ will denote the corresponding nonholonomic dynamics in the Hamiltonian side.

Let $\lambda$ be a 1-form on $Q$ such that $\lambda(Q)
\subseteq M$. Then, we can define a vector field on $Q$
\begin{equation}\label{eq:Def-Xnh-lambda}
	X_{nh}^{\lambda}=T\restr{(\pi_Q)}{M} \circ X_{nh}\circ \lambda\, .
\end{equation}

\begin{remark}\label{remark:X-nh-lambda}
	If $q\in Q$ and $f\in C^\infty(Q)$ then, from equations~\eqref{lift} and \eqref{eq:Def-Xnh-lambda} and Corollary~\ref{coro}, we deduce that
	\begin{align}
		\langle X_{nh}^\lambda (q), df (q) \rangle 
		& = \left\langle T_{\lambda(q)}\restr{(\pi_Q)}{M}\circ T_{\lambda(q)}\gamma\circ X_{H_M \circ\, \gamma}\circ \lambda(q), df(q)\right\rangle \\
		&= \left\langle X_{H_M \circ\, \gamma} (\lambda(q)), \pi_Q^\ast(df) (\lambda(q)) \right\rangle\\
		& = - \left(i_{(df(q))^V_{\lambda(q)}} \omega_Q(\lambda(q)) \right) \left(X_{H_M \circ\, \gamma} (\lambda(q)) \right)\\
		& = \restr{\frac{d}{dt}}{t=0} \left(H_M \circ\gamma\right)\left(\lambda(q) + t df(q) \right)\, .
	\end{align}
	Now, using that $\lambda(q)\in M$ (which implies that $\gamma \circ \lambda(q)= \lambda(q)$) and the definition of $H$ (see equation~\eqref{eq:H-mechanical}), we obtain that
	\begin{equation}
		\left\langle X_{nh}^\lambda(q), df (q) \right\rangle ) = \restr{\frac{d}{dt}}{t=0}H\left(\lambda(q) + t df(q) \right)\, ,
	\end{equation}
	which, from equation~\eqref{eq:key-point}, implies that
	\begin{equation}
		\left\langle X_{nh}^\lambda(q), df (q) \right\rangle
		= \left\langle X_{H}^\lambda(q), df (q) \right\rangle
		= \left\langle \sharp_g(\lambda(q)), df (q) \right\rangle\, .
	\end{equation}
	Thus, we conclude that 
	\begin{equation}\label{eq:X-nh-lambda}
		X_{nh}^\lambda(q) = \sharp_g(\lambda(q))\, ,
	\end{equation}
	as in the free case (see Remark~\ref{remark:mechanical_Hamiltonian}).
	In particular, since $\lambda(q)\in M_q = (D_q^{\perp_g})^\circ$, we have that
	\begin{equation}\label{eq:X-nh-lambda-D}
		X_{nh}^\lambda(q) \in D_q\, , 
	\end{equation}
	for all $q\in Q$. 
\end{remark}


Moreover, in \cite{Iglesias-Ponte2008} the authors proved the following result.

\begin{theorem}\label{hjprincipal}
Let $\lambda$ be a $1$-form on $Q$ taking values into $M$ and satisfying
$d\lambda \in {\cal I}(D^{\circ})$, where ${\cal I}(D^{\circ})$ denotes the ideal defined by $D^{\circ}$. Then the following conditions are
equivalent:

\begin{enumerate}
\item[(i)] $X_{nh}^\lambda$ and $X_{nh}$ are $\lambda$-related;

\item[(ii)] $d (H\circ \lambda)\in D^{\circ}$
\end{enumerate}

\end{theorem}

In consequence, the Hamilton--Jacobi equation for the nonholonomic system is
\begin{equation}\label{hj}
d ( H \circ \lambda) \in D^{\circ},
\end{equation}
assuming the additional conditions
\begin{equation}\label{hja}
\begin{aligned}
	& \lambda(Q) \subseteq M\, ,\\
	& d\lambda \in {\cal I}(D^{\circ})\, .
\end{aligned}
\end{equation}
Notice that $d\lambda \in {\cal I}(D^{\circ})$ if and only if 
\begin{equation}\label{hjb}
d\lambda (v_1, v_2) = 0 \, , 
\end{equation}
for all $v_1, v_2 \in D$ (see \cite{deLeon2010a,Ohsawa2009}).


We can improve the results in the above theorem when the distribution $D$ is completely nonholonomic (or bracket-generating),
that is, if $D$ along with all of its iterated Lie brackets $[D, D], [D, [D, D]], \ldots$ spans the tangent bundle $TQ$.

Indeed, using Chow's theorem, one can prove that if $Q$ is a connected differentiable manifold and 
$D$ is completely nonholonomic, then there is no non-zero exact one-form in the annihilator $D^{\circ}$. Therefore,
in this case $d( H \circ \lambda) \in D^{\circ}$ is equivalent to $d( H \circ \lambda) = 0$ (see \cite{deLeon2010a,Ohsawa2009}).


On the other hand, we can give a different proof of Theorem \ref{hjprincipal} using the properties of the the Eden bracket and some general results in \cite{deLeon2010a}. A sketch of this proof is the following one.

Using Theorem \ref{th1} and the fact that the almost Poisson bracket $\{\, , \, \}_{D^*}$ is linear on the vector bundle $D^*$ (see \cite{deLeon2010a}), we directly deduce that the
Eden bracket $\{\, , \, \}_E$ is also linear on the vector subbundle 
$M = (D^{\perp_g})^{\circ} \subseteq T^*Q$. So, $\{\, , \,\}_E$ induces an skew-symmetric algebroid structure
on the dual bundle $M^* = ((D^{\perp_g})^{\circ})^*$ (see Theorem 2.3 in \cite{deLeon2010a}).
Note that $M^*$ may be identified with the vector subbundle $D$. Indeed, the dual isomorphism
$$
i_{M, D^\ast}^* : D \to M^*
$$
to $i_{M,D^*} : M \to D^*$ is just an skew-symmetric algebroid isomorphism when on $D$ we consider the 
skew-symmetric algebroid structure $(\| \, , \, \|, i_D)$ induced by the linear almost Poisson bracket $\{\, , \, \}_{D^*}$. This structure 
$(\| \, , \, \|, i_D)$ was described at the beginning of Subsection~\ref{subsec:algebroid}. Now, using this description, and the general Theorem 4.1 in \cite{deLeon2010a}, we directly deduce Theorem \ref{hjprincipal}.

\subsection{A new formulation of the Hamilton--Jacobi theory for nonholonomic mechanical systems}\label{subsec:HJ_nh_new}

It is really interesting to express the projection $\gamma$ in bundle coordinates.
We can consider a basis $\{e_a\}$ of $\Gamma(D)$ and $\{\mu^A\}$ of $\Gamma(D^{\circ})$ such that
$$
e_a = e^i_a \, \frac{\partial}{\partial q^i} \; , \; \mu^A = \mu^A_i \, dq^i \, .
$$
As in the original papers by R.~Eden \cite{Eden1951,Eden1951b}, we can consider the regular matrix with components
$C_{ab} = g(e_a, e_b)$ and define 
${\cal E}^{kj} = e^k_a C^{ab} e^j_b$, where $C^{ab}$ are the components of the inverse matrix of $(C_{ab})$.
Then a direct computation shows that {the projector $\gamma$ defined in \eqref{eq:gamma} can be written as}
\begin{equation}\label{pes}
\gamma (q^i, p_i) = (q^i, \gamma^j_i p_j)\, ,
\end{equation}
where
$$
\gamma^j_i = g_{ik} {\cal E}^{kj} \, .
$$
Notice that $\gamma$ maps free state phases into constrained state phases, i.e.~points in $T^\ast Q$ into points in $M$. However, $\gamma$ does not map the free dynamics into the nonholonomic dynamics, i.e.~it does not map integral curves of $X_H$ into integral curves of $X_{nh}$. Nevertheless, $\gamma$ maps the free dynamics of a modified Hamiltonian into the nonholonomic dynamics (see Corollary~\ref{coro}).

With the above notations, one can see that equation~\eqref{hjb} can be locally written as
\begin{equation}\label{eq:HJ_ideal_coords}
	\left(\frac{\partial \lambda_l}{\partial q^k} - \frac{\partial \lambda_k}{\partial q^l}\right)e^k_a e^l_b = 0\, ,
\end{equation}
which is trivially satisfied if $\lambda = \lambda_i dq^i$ is closed.

On the other hand, the condition $\lambda(Q) \subseteq M$ can be locally written as
$$
\lambda_i = \gamma^j_i \lambda _j\, .
$$
Therefore, the solutions of the Hamilton--Jacobi equation for the nonholonomic system are 1-forms $\lambda\in \Omega^1(Q)$ satisfying the following conditions:
\begin{equation}\label{eq:HJ_nh}
\begin{aligned}
	&\lambda = \gamma\circ \lambda \, ,\\
	& \restr{\dd \lambda}{D\times D} = 0 \, ,\\
	&\gamma\circ  \dd(H\circ \lambda) = 0 \,,
\end{aligned}
\end{equation}
or, in bundle coordinates,
\begin{equation}\label{eq:HJ_nh_coords}
\begin{aligned}
	&\lambda_i = \gamma^j_i \lambda _j \, ,\\
	&\left(\frac{\partial \lambda_l}{\partial q^k} - \frac{\partial \lambda_k}{\partial q^l}\right) e^k_a e^l_b = 0 \, ,\\
	&\gamma^i_k \left(\frac{\partial H}{\partial q^i} + \frac{\partial H}{\partial p_j} \frac{\partial \lambda_j}{\partial q^i} \right) = 0 \,. 
\end{aligned}
\end{equation}
Observing the above equations, we can notice that if $\lambda$ is a solution for the unconstrained Hamilton--Jacobi problem (and $\lambda$ is assumed to be closed),
then $\lambda$ would be a solution for the nonholonomic Hamilton--Jacobi problem if and only if $\lambda$ takes values in $M$.

\subsection{Generalized nonholonomic Hamilton--Jacobi equation}\label{sec:HJ_algebroid}
In this section, we will proof a nonholonomic version of Theorem~\ref{thm:HJ-generalized}.

Assume that $(Q, g, V, D)$ is a nonholonomic mechanical system, and let $\lambda \in \Omega^1(Q)$ be a 1-form on $Q$ taking values in $M$, namely
$$
\lambda(Q) \subseteq M = (D^{\perp_g})^\circ\, .
$$
As above, we denote by $X_{nh}\in\mathfrak{X}(M)$ the nonholonomic dynamics in the Hamiltonian side, and by $X_{nh}^\lambda$ on $Q$ given by
\begin{equation}\label{eq:X-nh-lambda-r}
	X_{nh}^{\lambda} = \restr{T(\pi_Q)}{M} \circ X_{nh} \circ \lambda\, ,
\end{equation}
so that the following diagram commutes:

\begin{equation}
	\begin{tikzcd}
		M \arrow[r, "X_{nh}"]                              & TM \arrow[d, "\restr{T(\pi_Q)}{M} "] \\
		Q \arrow[r, "X_{nh}^\lambda"] \arrow[u, "\lambda"] & TQ                                    
	\end{tikzcd}
\end{equation}


As we know, 
\begin{equation}\label{eq:X-nh-D-q}
	X_{nh}^\lambda(q) = \sharp_g(\lambda(q)) \in D_q \, ,
\end{equation}
for every $q\in Q$ (see Remark~\ref{remark:X-nh-lambda}).

\begin{theorem} \label{thm:HJ-generalized-nonholonomic}
Let $\lambda\in \Omega^1(Q)$ such that $\gamma \circ \lambda = \lambda$. Then, the vector fields
$X_{nh}^\lambda$ and $X_{nh}$ are $\lambda$-related if and only if
\begin{equation}\label{eq:Gen-NH-HJ-Eq_gamma}
	\gamma \circ \big( \dd (H \circ \lambda) + i_{X_{nh}^\lambda} \dd \lambda \big) = 0\, .
\end{equation}

Equation~\eqref{eq:Gen-NH-HJ-Eq_gamma} will be called the \emph{generalized nonholonomic Hamilton--Jacobi equation}.
\end{theorem}

\begin{remark}\label{remark:gen_HJ_alternative}
Note that if $\alpha\colon Q\to TQ$ is a 1-form on $Q$ then it is easy to prove that
\begin{flalign}
	& \alpha(Q)\subseteq M \Leftrightarrow \gamma \circ \alpha = \alpha\,, \\
	& \gamma \circ \alpha = 0 \Leftrightarrow \alpha(v_q)=0\ \forall\, v_q\in D_q \text{ and } \forall\, q\in Q\, .
\end{flalign}
Thus, since $X_{nh}^\lambda(q)\in D_q$ for every $q\in Q$, we deduce that the generalized nonholonomic Hamilton--Jacobi equation \eqref{eq:Gen-NH-HJ-Eq_gamma} may be equivalently written as
\begin{equation}\label{eq:Gen-NH-HJ-Eq}
	d^D (H \circ \lambda) + i_{X_{nh}^\lambda} \, d^D \lambda = 0\, ,
\end{equation}
where $d^D$ is the pseudo-differential of the skew-symmetric algebroid $D$.	

We also remark the following facts, on results related with Theorem~\ref{thm:HJ-generalized-nonholonomic}, that one may find in the literature:
\begin{itemize}
	\item In \cite{Carinena2010}, the authors obtain a similar result but in the Lagrangian formulation.
	\item In \cite{Balseiro2010}, the authors discuss the Hamilton--Jacobi equation for nonholonomic mechanical systems subjected to affine nonholonomic constraints but in the skew-symmetric algebroid settting. The appearance of the Hamilton--Jacobi equation in \cite{Balseiro2010} is similar to equation~\eqref{eq:Gen-NH-HJ-Eq}, but the relevant space in \cite{Balseiro2010} is the affine dual of the constraint affine subbundle (which is different from our constraint vector subbundle $M$).
\end{itemize}

\end{remark}

In order to prove Theorem~\ref{thm:HJ-generalized-nonholonomic}, we will need the following lemmas.

\begin{lemma}\label{lemma:HJ-generalized-nonholonomic-1}
For every $q \in Q$, we have
$$
(T_q \lambda)(X_{nh}^\lambda(q)) \in T^D_{\lambda(q)} M
$$
\end{lemma}

\begin{proof}
	Since $\lambda (Q) \subseteq M$, we have that $(T_q\lambda)(X_{nh}^\lambda(q)) \in T_{\lambda(q)}M$.
	Moreover,
	\begin{eqnarray*}
	(T_{\lambda(q)} \restr{(\pi_Q)}{M})(T_q\lambda)(X_{nh}^\lambda(q)) &=& T_q(\restr{(\pi_Q)}{M} \circ \lambda)(X_{nh}^\lambda (q))\\
	&=& X_{nh}^\lambda (q) 
	\,.
	\end{eqnarray*}
	Thus, the result follows using equation~\eqref{eq:X-nh-D-q}.
\end{proof}

\begin{lemma}\label{lemma:HJ-generalized-nonholonomic-2}
For every $q \in Q$, we have
\begin{equation}\label{eq:Decom-T-D-M}
T^D_{\lambda(q)} M = (T_q\lambda)(D_q) \oplus V_{\lambda(q)}(\restr{\pi_Q}{M})\, .
\end{equation}
In addition, 
$$
V_{\lambda(q)} (\restr{\pi_Q}{M}) = \{(\beta_q)^V_{\lambda(q)} \; \mid \; \beta_q \in M_q = (D_q^{\perp_g})^\circ \}\, .
$$
\end{lemma}

\begin{proof}
It is easy to see that
$$
(T_q\lambda)(D_q) \cap V_{\lambda(q)} (\restr{\pi_Q}{M}) = \{0\}\, .
$$
Moreover, if $Z_{\lambda(q)} \in T_{\lambda(q)}^DM$ then
$$
Z_{\lambda(q)} = T_q \lambda \circ T_{\lambda(q)}\pi_Q \circ Z_{\lambda(q)}  + \big(Z_{\lambda(q)}  -T_q \lambda \circ T_{\lambda(q)}\pi_Q \circ Z_{\lambda(q)}\big)\, .
$$
In addition, we have
$$
T_{\lambda(q)} \pi_Q\circ Z_{\lambda(q)} \in D_q\, ,
$$
which implies that
$$
T_q \lambda \circ T_{\lambda(q)}\pi_Q \circ Z_{\lambda(q)}  \in (T_q \lambda)(D_q)\, .
$$
Furthermore, it is easy to see that
$$
Z_{\lambda(q)}  -T_q \lambda \circ T_{\lambda(q)}\pi_Q \circ Z_{\lambda(q)} \in V_{\lambda(q)} (\restr{\pi_Q}{M})\, .
$$
This proves equation~\eqref{eq:Decom-T-D-M}.

On the other hand, 
$$
V_{\lambda(q)} \pi_Q = \{(\beta_q)^V_{\lambda(q)} \; \mid \; \beta_q \in T^*_qQ \}\, ,
$$
and thus
$$
V_{\lambda(q)} (\restr{\pi_Q}{M}) = V_{\lambda(q)} \pi_Q  \cap TM
= \{(\beta_q)^V_{\lambda(q)} \; \mid \; \beta_q \in M_q = (D_q^{\perp_g}) ^\circ \} \, .
$$

\end{proof}

\begin{lemma}\label{lemma:HJ-generalized-nonholonomic-3}
	If $\beta_q\in M_q=(D_q^{\perp_g})^\circ$, then
	\begin{equation}\label{cuatro}
	(i_{(T_q\lambda)(X_{nh}^\lambda (q))} \, \omega_Q(\lambda(q)))(\beta_q)^V_{\lambda(q)} =
	(i_{X_{nh}(\lambda(q))} \, \omega_Q (\lambda(q)) (\beta_q)^V_{\lambda(q)}
	\end{equation}
	for all $\beta_q \in M_q$.
\end{lemma}

\begin{proof}
	Indeed, using equations~\eqref{lift} and \eqref{eq:X-nh-lambda-r}, we have
	\begin{eqnarray*}
	(i_{X_{nh}(\lambda(q))} \, \omega_Q (\lambda(q)) (\beta_q)^V_{\lambda(q)} 
	&=& - (i_{(\beta_q)^V_{\lambda(q)}} \, \omega_Q(\lambda(q)))(X_{nh}(\lambda(q)) \\
	&=& (T^*_{\lambda(q)} \pi_Q)(\beta_q)(X_{nh}(\lambda(q)))\\
	&=& \beta_q((T_{\lambda(q) }\pi_Q)(X_{nh}(\lambda(q))) \\
	&=&\beta_q(X_{nh}^\lambda(q)) \\
	&=& \langle (T^*_{\beta(q)}\pi_Q)(\beta_q), (T_q\lambda)(X_{nh}^\lambda (q)) \rangle \\
	&=& - (i_{(\beta_q)^V_{\lambda(q)}} \, \omega_Q(\lambda(q))(T_q \lambda) (X_{nh}^\lambda(q)) \\
	&=& (i_{(T_q\lambda)(X_{nh}^\lambda(q))} \, \omega_Q(\lambda(q))) (\beta_q)^V_{\lambda(q)}.
	\end{eqnarray*}
\end{proof}

We can now prove the theorem above.

\begin{proof}[Proof of Theorem~\ref{thm:HJ-generalized-nonholonomic}]

For every $q \in Q$, we have
$$
(T_q\lambda) (X_{nh}^\lambda (q)) = X_{nh}(\lambda(q)) \iff
i_{(T_q\lambda)(X_{nh}^\lambda (q))} \, \omega_Q (\lambda(q)) = i_{X_{nh}(\lambda(q))} \, \omega_Q (\lambda(q))\, .
$$
Thus, using Lemmas~\ref{lemma:HJ-generalized-nonholonomic-1}, \ref{lemma:HJ-generalized-nonholonomic-2} and \ref{lemma:HJ-generalized-nonholonomic-3} and the fact that
$$
T_{\lambda(q)} (T^*Q) = T_{\lambda(q)}^D M \oplus (T_{\lambda(q)}^D M)^{\perp_{\omega_Q}}\, ,
$$
we obtain that
\begin{equation}\label{eq:Clave-formula}
\begin{aligned}
	&(T_q\lambda)(X_{nh}^\lambda(q)) = X_{nh}(\lambda(q)) \iff \\
	&(i_{T_q\lambda (X_{nh}^\lambda(q))} \, \omega_Q (\lambda(q)))(T_q\lambda)(u_q)
	= (i_{X_{nh}(\lambda(q)} \, \omega_Q (\lambda(q)))(T_q\lambda)(u_q), \hbox{ for all} \; u_q \in D_q
\end{aligned}
\end{equation}


Now, if $\theta_Q$ is the Liouville 1-form of $T^\ast Q$, then it is well-known that $\lambda^\ast \theta_Q = \theta_Q$ (see, for instance, \cite{Abraham2008}). Using this fact, we deduce that 
\begin{align}
	\left( i_{T_q\lambda (X_{nh}^\lambda(q))} \omega_Q(\lambda(q)) \right) \left( T_q\lambda (u_q) \right)
	&= -\left[\lambda^\ast (d \theta_Q)\right] (q) \left(X_{nh}^\lambda(q), u_q\right)\\
	&= -d\lambda(q) \left(X_{nh}^\lambda(q), u_q\right)
	= - \left(i_{X_{nh}^\lambda} d\lambda\right) (q) (u_q)\, .
\end{align}
Taking into account that $X_{nh}^\lambda(q) \in D_q$ (see Remark~\ref{remark:X-nh-lambda}), it follows that 
\begin{equation}\label{eq:First-formula}
	\left( i_{T_q\lambda(X_{nh}^\lambda(q))} \omega_Q(\lambda(q)) \right) \left( T_q\lambda (u_q) \right)
	= - \left(i_{X_{nh}^\lambda} d^D\lambda\right) (q) (u_q)\, .
\end{equation}
On the other hand, from equation~\eqref{eq:Int-Def-X-nh} and since $u_q\in D_q$, we obtain that
\begin{equation}\label{eq:second-formula}
	\left( i_{X_{nh}(\lambda(q))} \omega_Q(\lambda(q)) \right) \left( T_q\lambda (u_q) \right)
	= \left[d(H\circ \lambda)(q) \right] (u_q) = \left[d^D(H\circ \lambda)(q) \right] (u_q)\, .
\end{equation}
Finally, using Remark~\ref{remark:gen_HJ_alternative} and equations~\eqref{eq:Clave-formula}, \eqref{eq:First-formula} and \eqref{eq:second-formula} we deduce the result.




\end{proof}

Following the same notation as in Subsection~\ref{subsec:HJ_nh_new}, proceeding as in that subsection and using the fact that $X_{nh}^\lambda= \sharp_g(\lambda)$, we deduce that a 1-form $\lambda=\lambda_i d q^i \in \Omega^1(Q)$ satisfies the condition $\lambda(Q)\subseteq M$ and the generalized nonholonomic Hamilton--Jacobi equation~\eqref{eq:Gen-NH-HJ-Eq_gamma} if and only if
\begin{equation}
	\lambda_i = \gamma_i ^j \lambda_j, \quad \text{for all } i
\end{equation}
and
\begin{equation}\label{eq:Gen-NH-HJ-Eq_coords}
	\gamma^i_j \left( \left(\frac{\partial H}{\partial q^j}+\frac{\partial \lambda_k}{\partial q^j} \frac{\partial H}{\partial p_k}\right) 
	+g^{kl}\lambda_l \left(\frac{\partial \lambda_k}{\partial q^j} - \frac{\partial \lambda_j}{\partial q^k} \right)
	\right)=0\, .
\end{equation}


\section{Examples} \label{sec:example}
\subsection{The nonholonomic particle}

Consider a particle of unit mass be moving in space $Q = \mathbb{R}^3$, with Lagrangian
$$
L = \frac{1}{2} (\dot{x}^2+\dot{y}^2+\dot{z}^2) - V(x,y,z)\, ,
$$
and subject to the constraint
$$
\Phi = \dot{z} - y \dot{x} = 0\, .
$$
Following the previous notations, this means that the distribution $D$ is generated by
the global vector fields
$$
 e_1 = \frac{\partial}{\partial y},\quad  e_2 = \frac{\partial}{\partial x} + y \frac{\partial}{\partial z} \, .
$$
Moreover, we have
$$
D^{\perp_g} = \left\langle \frac{\partial}{\partial z} - y \frac{\partial}{\partial x} \right\rangle\, ,
$$
and
$$
D^{\circ} = \langle dz - y dx \rangle\, .
$$
Passing to the Hamiltonian side, we obtain the Hamiltonian function
$$
H(x,y,z,p_x,p_y,p_z)= \frac{1}{2}(p_x^2 +p_y^2 +p_z^2)+V(x,y,z)\, ,
$$
with constraints given by the function
$$
\Psi = p_z - y p_x = 0\, .
$$
We have an orthogonal decomposition
$$
T^*Q = M \oplus D^{\circ}\, ,
$$
where a simple computation shows that
$$
M = \langle dy, dx + y dz \rangle\, .
$$
Thus, we can take global coordinates $(x, y, z, \pi_1, \pi_2)$ on $M$, and using equation~\eqref{localexpression-E-bracket} we obtain the following equations for the Eden bracket:
\begin{eqnarray*}
\{x, \pi_2\}_E &=& - \{\pi_2, x\}_ E = 1\, ,\\
\{y, \pi_1\}_E &=& - \{\pi_1, y\}_ E = 1\, ,\\
\{z, \pi_2\}_E &=& - \{\pi_2, z\}_ E = y\, ,
\end{eqnarray*}
the rest of the brackets between the coordinates being zero.

Next, a straightforward calculation shows that
\begin{align}
	&g= \begin{pmatrix}
    1 & 0 & 0\\
    0 & 1 & 0\\
    0 & 0 & 1
	\end{pmatrix}\, , \\\\
	&C^{-1} = \begin{pmatrix}
		1 & 0 \\
		0& \frac{1}{1 + y^2}
	\end{pmatrix}\, ,\\\\
	&\mathcal{E} = 
	\begin{pmatrix}
		\frac{1}{1+y^2} & 0 & \frac{y}{1+y^2}\\
		0 & 1 & 0\\
    	\frac{y}{(1+y^2} & 0 & \frac{y^2}{1+y^2}
	\end{pmatrix}\, ,
\end{align}
and $\gamma \equiv \mathcal{E}$.

Hence, if $p_x dx + p_y dy + p_z dz$ is a 1-form on $\mathbb{R}^3$ then
$\tilde{p}_x dx + \tilde{p}_y dy + \tilde{p}_z dz = \gamma (p_x dx + p_y dy + p_z dz) \in \Gamma(M)$, and we have that
\begin{eqnarray*}
\tilde{p}_ x & = & \frac{1}{1+y^2}p_x + \frac{y}{1+y^2}p_z\, , \\
\tilde{p}_y & = & p_y\, ,\\
\tilde{p}_z &=& \frac{y}{1+y^2} p_x  + \frac{y^2}{1+y^2} p_z\, .
\end{eqnarray*}
Let $\lambda\in \Omega^1(Q)$ be a solution of the Hamilton--Jacobi equation~\eqref{eq:HJ_nh}. The condition $\gamma \circ \lambda = \lambda$ implies that $\lambda$ is of the form
$$\lambda = \lambda_x \dd x + \lambda_y \dd y + y \lambda_x \dd z\, .$$
On the other hand, the condition $\restr{\dd \lambda}{D\times D} = 0$ holds if and only if
$$\dd \lambda(e_1, e_2) = 0\, ,$$
or, equivalently,
$$\left(1+y^2\right) \frac{\partial \gamma_x}{\partial y}+y \gamma_x-\frac{\partial \gamma_y}{\partial x}-y \frac{\partial \gamma_y}{\partial z}=0\, .$$
The Hamilton--Jacobi equation $\gamma \circ \dd(H\circ \lambda) = 0$ yields
\begin{align}
	& \lambda_x \frac{\partial \lambda_x}{\partial x}+\lambda_y \frac{\partial \lambda_y}{\partial x}+y^2 \lambda_x \frac{\partial  \lambda_x}{\partial x}+y\left(\lambda_x \frac{\partial \lambda_x}{\partial z}+\lambda_y \frac{\partial \lambda_y}{\partial z}+y^2 \lambda_x \frac{\partial \lambda_x}{\partial z}\right)+\frac{\partial V}{\partial x}+y \frac{\partial V}{\partial z}=0\, ,\\
	& \lambda_x \frac{\partial \lambda_x}{\partial y}+\lambda_y \frac{\partial \lambda_y}{\partial y}+y^2 \lambda_x \frac{\partial  \lambda_x}{\partial y}+y \lambda_x^2 +\frac{\partial V}{\partial y}=0\, .
\end{align}
These equations coincide with those obtained in Example 6.1 from \cite{deLeon2014}.

In particular, if the Hamiltonian is purely kinetical (i.e.~$V=0$), then a solution for the Hamilton--Jacobi equation is given by
$$\lambda=\frac{\mu}{\sqrt{1+y^2}} \dd x \pm \sqrt{2 E-\mu^2} \dd y+\frac{\mu y}{\sqrt{1+y^2}} \dd z\, ,$$
for some constants $E$ and $\mu$ (see Example 5.3.1 in \cite{Colombo2022}).

\subsection{The rolling ball}
Consider a sphere of radius $r$ and mass $1$ which rolls without sliding on a horizontal plane. The configuration space is $Q=\RR^2 \times \mathrm{SO}(3)$.
The Lagrangian function $L\colon TQ\to \RR$ is given by
$$L=\frac{1}{2}m\left(\dot{x}^2+\dot{y}^2\right) + \frac{1}{2} \langle \mathbb{I} \omega, \omega \rangle\,, $$
where $\omega=(\omega_1, \omega_2, \omega_3)$ denotes the angular velocity of the ball, and $\mathbb{I}$ the moment of inertia of the sphere with respect to its center of mass. Assume that the sphere is homogeneous. Then, $\mathbb{I}=\mathrm{diag}(I, I, I)$.

The ball rotates without sliding, i.e.~its subject to the nonholonomic constraints
$$\dot x = r \omega_2, \quad \dot y = -r \omega_1\,. $$
Let $X_1^R, X_2^R, X_3^R$ denote the standard basis of right-invariant vector fields on $\mathrm{SO}(3)$. Let $\rho_1, \rho_2, \rho_3$ be the right Maurer--Cartan 1-forms, which form a basis of $T^\ast \mathrm{SO}(3)$ dual to $\{X_1^R, X_2^R, X_3^R\}$. Then, the constraint 1-forms are 
$$\mu^1 = \dd x - r\rho_2,\qquad 
\mu^2 = \dd y + r \rho_1\, , $$
which span $D^\circ$. 
Hence, $\Gamma(D)=\langle e_a, e_b, e_c\rangle$ and $\Gamma(D^\perp)=\langle e_\alpha, e_\beta\rangle$, where
$$ e_a 	= \frac{\partial}{\partial x}+\frac{1}{r}X_2^R,\quad e_b = \frac{\partial}{\partial y}-\frac{1}{r}X_1^R,\quad e_c =  X_3^R, \quad e_\alpha = I \frac{\partial}{\partial x} - mr X_2^R, \quad e_\beta = I \frac{\partial}{\partial y} + mr X_1^R\, .$$
Their brackets are given by
$$ [e_a, e_b] = -\frac{1}{r^2} e_b,\quad [e_a, e_c] = \frac{I}{I+mr^2} e_b - \frac{1}{I+mr^2} e_\beta, \quad
[e_b,e_c]=-\frac{I}{I+mr^2} e_a - \frac{1}{I+mr^2} e_\alpha \, .$$
From the orthogonal basis $\{e_a, e_b, e_c, e_\alpha, e_\beta\}$ and using the Euler angles $( \theta, \varphi, \psi)$ as coordinates of $\mathrm{SO(3)}$, we have local coordinates $(x,y, \theta, \varphi, \psi, v^a, v^b, v^c, v^\alpha, v^\beta)$ in $TQ$, where
$$\dot x = v^a + I v^\alpha, \quad \dot y = v^b + I v^\beta, \quad \omega_1 = - \frac{v^b}{r} + mr v^\beta, \quad \omega_2 = \frac{v^a}{r} - mr v^\alpha, \quad \omega_3= v^c\, .$$ 
In these new coordinates, the Lagrangian function is given by
$$ L = \frac{1}{2}m\left[(v^a + I v^\alpha)^2+(v^b + I v^\beta)^2\right] 
	+ \frac{I}{2} \left[ \left(- \frac{v^b}{r} + mr v^\beta\right)^2 + \left(\frac{v^a}{r} - mr v^\alpha\right)^2 +\left(v^c\right)^2  \right]\,. $$
Since it is purely kinetical, the Lagrangian energy coincides with the Lagrangian function, namely, $E_L=L$.
The constraint submanifold $D\subseteq TQ$ is given by
$$D=\left\{ (x,y, \theta, \varphi, \psi, v^a, v^b, v^c, v^\alpha, v^\beta) \mid v^\alpha=v^\beta=0\right\}\, ,$$
so the canonical inclusion $i_D\colon D \hookrightarrow TQ$ is 
$$i_D\colon (x,y, \theta, \varphi, \psi, v^a, v^b, v^c)\mapsto (x,y, \theta, \varphi, \psi, v^a, v^b, v^c, 0,0)\, .$$
The constrained Lagrangian function is
$$ L \circ i_D =  \frac{1}{2}\left(m+ \frac{I}{r^2}\right) \left[(v^a) ^2+(v^b) ^2\right] 
+ \frac{I}{2} \left(v^c\right)^2 \,. $$
The Legendre transformation and its inverse are given by
\begin{align}
	FL&\colon
	\left(x, y, \theta, \varphi, \psi, v^a, v^b, v^c, v^\alpha, v^\beta\right) \\ 
	\\& \mapsto
	\left(x, y, \theta, \varphi, \psi, \frac{I+mr^2}{r^2}v^a, \frac{I+mr^2}{r^2}v^b, Iv^c, Im (I+mr^2) v^\alpha, Im (I+mr^2)  v^\beta\right)\, ,
\end{align}
and 
\begin{align}
	FL^{-1}&\colon
	\left(x, y, \theta, \varphi, \psi, p_a, p_b, p_c, p_\alpha, p_\beta\right) \\ 
	\\& \mapsto
	\left(x, y, \theta, \varphi, \psi, \frac{r^2}{I+mr^2}p_a, \frac{r^2}{I+mr^2}p_b, \frac{p_c}{I}, \frac{p_\alpha}{Im (I+mr^2)}, \frac{p_\beta}{Im (I+mr^2)}\right)\, ,
\end{align}
respectively.
The Hamiltonian function $H\colon T^\ast Q \to \RR$ is 
$$H = E_L \circ FL^{-1}
= \frac{r^2 p_a^2}{2 \left(I+m r^2\right)}+\frac{r^2 p_b^2}{2 \left(I+m r^2\right)}+\frac{p_c^2}{2 I}
+ \frac{p_\alpha^2}{2Im(I+mr^2)}+\frac{p_\beta^2}{2Im(I+mr^2)}
\, .$$
The constraint submanifold $M=(D^{\perp_g})^\circ \subseteq T^\ast Q$ is given by
$$M = \left\{\left(x, y, \theta, \varphi, \psi, p_a, p_b, p_c, p_\alpha, p_\beta\right) \mid p_\alpha=p_\beta=0\right\}\, ,$$
so the canonical inclusion $i_M\colon M\hookrightarrow T^\ast Q$  is
$$i_M \colon \left(x, y, \theta, \varphi, \psi, p_a, p_b, p_c\right) \mapsto \left(x, y, \theta, \varphi, \psi, p_a, p_b, p_c, 0, 0\right)\, .$$
Thus, the constrained Hamiltonian function is
$$ H \circ i_M = \frac{r^2 p_a^2}{2 \left(I+m r^2\right)}+\frac{r^2 p_b^2}{2 \left(I+m r^2\right)}+\frac{p_c^2}{2 I}\, .$$
Let $\{\mu^a, \mu^b, \mu^c, \mu^\alpha, \mu^\beta\}$ denote the dual basis of $\{e_a, e_b, e_c, e_\alpha, e_\beta\}$ . Then,
\begin{align}
	&\mu^a = \frac{r}{I+mr^2} (mr \dd x + I \rho_2),\quad \mu^b = \frac{r}{I+mr^2} (mr \dd y - I \rho_1), \quad \mu^c = \rho_3,\\
	& \mu^\alpha = \frac{1}{I+mr^2}(\dd x -r \rho_2), \quad \mu^\beta= \frac{1}{I+mr^2}(d y + r \rho_1) \, .
\end{align} 
The constrained Legendre transformation $F(L\circ i_D)\colon D \to M$ is
$$
	F(L\circ i_D) \colon
	\left(x, y, \theta, \varphi, \psi, v^a, v^b, v^c\right)  \mapsto
	\left(x, y, \theta, \varphi, \psi, \frac{I+mr^2}{r^2}v^a, \frac{I+mr^2}{r^2}v^b, Iv^c\right)\, ,
$$
and its inverse is
$$ 
	F(L\circ i_D)^{-1}\colon
	\left(x, y, \theta, \varphi, \psi, p_a, p_b, p_c\right) \mapsto
	\left(x, y, \theta, \varphi, \psi, \frac{r^2}{I+mr^2}p_a, \frac{r^2}{I+mr^2}p_b, \frac{p_c}{I}\right)\, .
$$

Let us now look for a solution $\lambda\in \Omega^{1}(Q)$ for the generalized nonholonomic Hamilton--Jacobi equation~\eqref{eq:Gen-NH-HJ-Eq}. The condition $\lambda(Q)\subseteq M$ implies that $\lambda$ is of the form
$$ \lambda = \lambda_a\, \mu^a + \lambda_b\, \mu^b + \lambda_c\, \mu^c\, ,$$
for some functions $\lambda_a, \lambda_b, \lambda_c \colon Q \to \RR$. Then, $\lambda$ is a solution of the generalized nonholonomic Hamilton--Jacobi equation if and only if
$$
d^D (H \circ \lambda) + i_{X_{nh}^\lambda} \, d^D \lambda = 0\, ,
$$
where $d^D$ denotes the pseudo-differential of the skew-symmetric algebroid $D$.
For simplicity's sake, assume that $d^D (H \circ \lambda)=0$ and $i_{X_{nh}^\lambda} \, d^D \lambda=0$.
We have that
$$ d^D (H \circ \lambda) = \frac{r^2 \lambda_a }{I+m r^2} \dd \lambda_a +\frac{r^2\lambda_b}{I+m r^2} \dd \lambda_b +\frac{\lambda_c}{I} \dd \lambda_c\, ,$$
which vanishes if $\lambda_a=c_a,\, \lambda_b=c_b$ and $\lambda_c = c_c$ for some constants $c_a,c_b,c_c\in \RR$. Then, we have that
\begin{align}
	& d\lambda\left(e_a, e_b\right)=-\lambda\left[e_a, e_b\right] = \frac{c_c}{r^2}\, ,\\
	& d\lambda\left(e_a, e_c\right)=-\lambda\left[e_a, e_c\right] = - \frac{Ic_b}{I+mr^2}\, , \\
	& d\lambda\left(e_b, e_c\right)=-\lambda\left[e_b, e_c\right] = \frac{Ic_a}{I+mr^2}\, ,
\end{align}
and thus
\begin{equation}\label{eq:dDlambda_coords}
	d^D \lambda=\frac{c_c}{r^2} \mu^a \wedge \mu^b-\frac{I c_b}{I+m r^2} \mu^a \wedge \mu^c+\frac{I c_a}{I+m r^2} \mu^b \wedge \mu^c \, .
\end{equation}
On the other hand, we have that
$$X_{nh}^\lambda(q) = \sharp_g (\lambda(q))\, ,$$
for each $q\in Q$, where $\sharp_g\colon T_q^\ast Q \to T_qQ$ denotes the isomorphism defined by the Riemannian metric $g$. Hence,
\begin{equation}\label{eq:Xnhlambda_coords}
	X_{n h}^\lambda=\frac{c_a r^2}{I+m r^2} e_a+\frac{c_b r^2}{I+m r^2} e_b+\frac{c_c}{I} e_c \, .
\end{equation}
Making use of the local expressions \eqref{eq:dDlambda_coords} and \eqref{eq:Xnhlambda_coords}, we obtain that
$$i_{X_{nh}^\lambda} d^D \lambda = 0\, ,$$
and conclude that $\lambda$ is a solution of the generalized nonholonomic Hamilton--Jacobi equation.

It is worth remarking that from this solution one can obtain 3 independent first integrals of the nonholonomic dynamics. 
Indeed, 
the map $\psi\colon Q \times \RR^3 \to M$ given by
$$\psi\colon (q, c_a, c_b, c_c) \mapsto c_a \mu^a(q) + c_b \mu^b(q) + c_c \mu^c(q)$$
is a global trivialization of $M$. Its inverse is given by 
$$\psi^{-1}\colon M\ni (q, p_a, p_b, p_c) \mapsto (q, p_a, p_b, p_c) \in Q\times \RR^3
\, .$$
Define the functions $f_a, f_b, f_c\colon M \to \RR$ given by 
$$f_a=p_a,\quad f_b=p_b,\quad f_c=p_c\, .$$  
Using equations~\eqref{Local-Poisson-bracket} and \eqref{lift}, we have
$$\left\{f_a, f_b\right\}_E=\frac{f_c}{r^2},\quad
\left\{f_c, f_a\right\}_E=\frac{I}{I+m r^2} f_b,\quad
\left\{f_b, f_c\right\}_E=\frac{I}{I+m r^2} f_a\, ,$$
so
$$\left\{f_a, H \circ i_M\right\}_E=\left\{f_b, H_{\circ} i_M\right\}_E=\left\{f_c, H \circ i_M\right\}_E=0\, ,$$
and thus $f_a, f_b$ and $f_c$ are first integrals of the nonholonomic dynamics.

Translating these first integrals to the Lagrangian formalism, we obtain that the functions
$$
	f_a \circ F(L\circ i_D) =  \frac{I+mr^2}{r^2}v^a,\quad  f_b\circ  F(L\circ i_D) = \frac{I+mr^2}{r^2}v^b,\quad f_c\circ  F(L\circ i_D)=Iv^c
$$
are first integrals for the nonholonomic Lagrangian dynamics. Hence, $v^a, v^b$ and $v^c$ are also first integrals for the nonholonomic Lagrangian dynamics. Therefore, $\omega_1,\, \omega_2$ and $\omega_3$ are first integrals as well. As a matter of fact, they coincide with the first integrals obtained in \cite[pp.~194-198]{Neimark2004}.

\section{Conclusions and future work}\label{sec:conclusions}

We have presented the concept of the Eden bracket and contrasted it with other nonholonomic brackets. It is noteworthy that there exist almost Poisson isomorphisms among the three nonholonomic mechanics formulations. Hence, one can make use of the formulation that is more convenient for each problem, and translate it to the other formulations via these isomorphisms.

The use of this new description of the nonholonomic bracket following Eden's ideas opens up many possibilities to simplify some developments in nonholonomic mechanics, including the following:


\begin{itemize}

\item We are going to study the quantization of nonholonomic systems \cite{deLeonc}. More specifically, following the original ideas by Eden \cite{Eden1951b}, we plan to study what is the quantum counterpart of a mechanical system with nonholonomic constraints.


\item We would also like to discuss the connection between complete solution of the generalized nonholonomic Hamilton--Jacobi equation, complete systems of first integrals of the nonholonomic system and symmetries of the system. In addition, the Eden bracket could be used to study of the reduction by symmetries and define a new version of the nonholonomic momentum map.

\item Moreover, we plan to construct a discrete version of the operator $\gamma$ in order to develop geometric integrators for nonholonomic mechanical systems.

\end{itemize}

\subsection*{Data availability}
Data sharing is not applicable to this article as no new data were created or
analyzed in this study.

\subsection*{Declaration of interest}
The authors have no competing interests to declare.

\section*{Acknowledgements}
The authors acknowledge financial support from Grant RED2022-134301-T funded by MCIN/AEI/ 10.13039/501100011033.
M.~de León, M.~Lainz and A.~López-Gordón also acknowledge Grants PID2019-106715GB-C21, PID2022-137909NB-C21 and CEX2019-000904-S funded by MCIN/AEI/ 10.13039/501100011033. Asier L\'opez-Gord\'on would also like to thank MCIN for the predoctoral contract PRE2020-093814. J.~C.~Marrero ackowledges financial support from the Spanish Ministry of Science and Innovation and European Union (Feder) Grant PID2022-137909NB-C22. He also thanks L.~García-Naranjo for his useful comments on first integrals of the rolling ball. All the authors would like to express their appreciation to the referee for their valuable feedback and constructive comments, which greatly improved the clarity of this paper.

\let\emph\oldemph


\printbibliography

\end{document}


		\begin{center}
		\begin{tikzcd}
    & T^\ast Q \arrow[dd, "\mu^A"] \\
D \arrow[rd, "\bar{\mu}^A"] \arrow[ru, "i_{D}", hook] &                              \\
    & \mathbb{R}                  
\end{tikzcd}
\end{center}

	\begin{corollary}
	We have
	$$
	{\cal P}(X_{H \circ \gamma}) = X_{nh}
	$$
	\end{corollary}
	
	\begin{proof}
	Indeed, for any function $f \in C^\infty(M)$, we have
	\begin{eqnarray*}
	X_{nh}(f) &=& \dot{f} \\
	&=& \{f, H_M\}_E \\
	&=& \{f \circ \gamma, H \circ \gamma\}\\
	&=& \omega_Q(X_{f \circ \gamma}, X_{H \circ \gamma}\} \\
	&=& omega_Q(X_{f \circ \gamma}, {\cal P}(X_{H \circ \gamma})\} \\
	&=& (i_{X_{f\circ \gamma}} \omega_Q) {\cal P}(X_{H \circ \gamma})\\
	&=& {\cal P}(X_{H \circ \gamma}) (f)
	\end{eqnarray*}
	since $f \circ \gamma^2 = f \circ \gamma$.
	\end{proof}

	We can consider an adapted orthonormal basis with respect to the above decomposition, that is, a local basis
	$\{e_i\} = \{e_a, e_A\}$
	of vector fields on $Q$ such that $\{e_a\}$ is a local basis of $D$ and $\{e_A\}$ is a local basis of $D^{\perp_g}$; in addition, we have
	\begin{eqnarray*}
	 && g(e_a, e_b) = \delta_{ab}\\
	 && g(e_a, e_B) = 0\\
	 && g(e_A, e_B) = \delta_{AB}
	\end{eqnarray*}